\documentclass[10pt]{article}
\usepackage{graphicx}
\usepackage{url}
\usepackage{amsmath}
\usepackage{amsthm}
\usepackage{amssymb}
\usepackage{algorithm,algorithmic}
\usepackage{epsfig}
\usepackage{color}
\newtheorem{proposition}{Proposition}

\providecommand{\norm}[1]{\lVert#1\rVert}
\begin{document}
\title{Analysis of Multiple Flows using Different High Speed TCP protocols on a General Network}
\author{Sudheer Poojary \qquad Vinod Sharma \\ Department of ECE, Indian Institute of Science, Bangalore \\ Email: \{sudheer, vinod\}@ece.iisc.ernet.in}
\maketitle
\begin{abstract}
We develop analytical tools for performance analysis of multiple TCP flows (which could be using TCP CUBIC, TCP Compound, TCP New Reno) passing through a multi-hop network. We first compute average window size for a single TCP connection (using CUBIC or Compound TCP) under random losses. We then consider two techniques to compute steady state throughput for different TCP flows in a multi-hop network. In the first technique, we approximate the queues as M/G/1 queues. In the second technique, we use an optimization program whose solution approximates the steady state throughput of the different flows. Our results match well with ns2 simulations.
\end{abstract}

\section{Introduction}
\label{sec:introduction}
TCP ensures reliable end-to-end data transfer between hosts and also does flow control and congestion control. While the goal of flow control is to avoid overwhelming the receiver with more packets than it can handle, the role of congestion control is to avoid congestion over the network. A good congestion control scheme must avoid congestion collapse, i.e., overwhelming congestion in the network with severely degraded throughput for all users. It must also be fair and efficient. 

Traditional TCP variants (TCP Reno, TCP New Reno) have been very effective in preventing Internet congestion collapse for more than two decades \cite{Huston06}. However as link speeds increase, these TCP variants have been found to be inefficient. The traditional TCP variants use AIMD (additive increase multiplicative decrease) algorithm for congestion window evolution. While AIMD is fair in the sense that all flows going through the same set of links get the same throughput \cite{Chiu1989Analysis}, it need not be efficient, especially for high bandwidth-delay product scenarios. The primary reason for the inefficiency being the slow rate of increase of window size when there is no congestion and the drastic multiplicative drop when there is congestion. Recent years have brought forth high speed variants of TCP (e.g., TCP CUBIC \cite{Ha2008} , FAST TCP \cite{Jin05}, H-TCP \cite{Leith2004}, Compound TCP \cite{Tan2006Infocom}). These alter the `AI' phase to make it more aggressive so that links are rarely underutilized. The high-speed TCP variants are now commonly used. TCP CUBIC has been in use on Linux systems since 2006 (Linux kernel 2.6.16) and TCP Compound is used for Windows server. Consequently, measurements in \cite{Yang2014} on $30000$ web servers reveal that a majority of these web servers ($ > 60 \%$) use either TCP CUBIC, TCP BIC (the predecessor of CUBIC) or TCP Compound whereas usage of Reno is restricted to less that $15 \%$.

Along with rapid rise in link speeds, we see a rapid increase in the number of mobile devices accessing the Internet. Wireless links are prone to fading, interference and attenuation which may cause random packet losses. Since TCP treats losses as indication of congestion and reduces its window size in response to packet loss, random packet losses considerably deteriorate TCP throughput.

In this paper we develop models for TCP performance in a network with multiple queues and multiple TCP flows (these may use Compound, CUBIC or New Reno) with random packet losses. We will be interested in the performance of long-lived flows (e.g., HTTP streaming, backups, large file downloads).  The slow start phase of TCP is of interest when studying performance of short-lived flows but does not impact the performance of long-lived flows if they have low packet rates $(< 1\%)$. Since our focus is on performance of the different TCP variants mentioned above, which differ only in the congestion avoidance phase, we will not consider the slow start phase.

\subsection{High speed TCP variants: A brief overview}
High speed TCP variants alter the congestion avoidance phase behaviour of the traditional TCP to achieve higher efficiency. High speed TCP (HSTCP) \cite{rfc3649} addresses the issue of low link utilization in large bandwidth-delay product (BDP) networks by making the window increments and decrements a function of the current window size. Scalable TCP \cite{Kelly2003Scalable} replaces additive increase by multiplicative increase so that the TCP flow rate converges quickly to the link speeds in large BDP networks. H-TCP \cite{Leith2004} alters the window increments so that increments depend on time since last congestion. BIC congestion control \cite{Xu04} uses binary search to obtain an efficient operating point. FAST TCP \cite{Wei2006} is different from the earlier described variants in the sense that it uses queuing delay as a measure of congestion unlike the other variants which use packet loss.  It can be considered to be a high speed variant of TCP Vegas\cite{Brakmo1995}. FAST TCP uses variable size increments and decrements based on the estimate of queuing in the network. In \cite{Ha2008}, the authors propose TCP CUBIC congestion control algorithm. Like H-TCP, TCP CUBIC window size is a function of time elapsed since last congestion. The authors choose a cubic function for window evolution. TCP Compound \cite{Tan2006Infocom} is a delay-based congestion control algorithm. Delay-based congestion control algorithms experience lower losses, lower queuing delays and better RTT-fairness as compared to loss-based congestion control algorithms \cite{Wei2006}, \cite{Bonald1999}. However, in a mixed environment, when competing with flows using loss-based congestion control, they are not able to get their fair share of capacity. To address this drawback, TCP compound window also has a loss-based component which gives it a worst-case performance of TCP New-Reno. An extensive survey describing the different TCP variants can be found in \cite{Afanasyev10}. 

There is a considerable amount of literature on simulation and experimental evaluation of the high speed TCP variants. In \cite{Jain2011}, the authors perform experimental evaluation of TCP CUBIC in a small buffer regime.  Results in \cite{Bateman2008} show that the ns2 implementation results for several high speed TCP variants match well with experimental results. The intra-protocol and inter-protocol fairness of different high speed TCP variants has been studied in \cite{Weigle2006} and \cite{Molnar2009} using ns2 simulations. The Reference \cite{Xue2014} is a recent experimental survey of high speed TCP fairness. It evaluates fairness of TCP SACK, HSTCP and CUBIC TCP on a $10$ Gbps optical link. They show that fairness is a function of buffer sizes and queue management schemes at the routers with RED routers yielding more equitable rate allocations than droptail.

\subsection{Analytical studies of TCP}
Traditional TCP (TCP Tahoe, TCP Reno and TCP New Reno) has been extensively studied and analyzed. In \cite{Mathis1997}, the authors use a periodic loss model to compute the average window size of TCP as a function of the packet error probability. The reference \cite{Padhye2000} uses Markov regenerative processes to model the window evolution of TCP Reno (congestion avoidance phase) under random losses. In \cite{Cardwell2000}, the authors consider the effect of connection establishment, slow start and congestion avoidance phase on TCP latency. In \cite{Sharma2002}, the authors compute the throughput and mean sojourn time of TCP Tahoe and TCP Reno flows over a single bottleneck link using RED queue management. The link also carries UDP traffic which has priority over the TCP traffic. In \cite{Sharma2004}, the authors prove stability of multiple TCP Tahoe and TCP Reno flows passing through a single drop-tail or RED queue in the presence of UDP traffic. They then extend their results to the situation when there are multiple bottleneck links in tandem. In \cite{Gupta2006}, the authors use Markovian models to compute mean download times for ON-OFF TCP Tahoe and TCP Reno flows and throughput for long-lived TCP flows in the presence of UDP traffic.  

An optimization-based approach is used to analyze network congestion control in \cite{Kunniyur2003, Hurley1999, Low2003, Massoulie2002, Altman2002}. In \cite{Kunniyur2003}, the authors consider a generic network-wide global optimization problem and derive a distributed congestion control algorithm whose equilibrium rate allocations are a solution to the global optimization problem. As opposed to this approach, \cite{Hurley1999} and \cite{Low2003} start with distributed congestion control algorithms and derive the corresponding network-wide global optimization problem. In \cite{Massoulie2002}, under a AIMD TCP-like scheme with fixed window size, the authors show that FIFO queuing gives proportional-fair rate allocation, longest queue first gives maximum sum-rate and fair queuing policy yields max-min fair rate allocation. An analytical model for identifying the bottleneck links in a multi-hop network is given in \cite{Altman2002}. 


Differential equations are used to model TCP behaviour in \cite{Lakshman1997, Misra2000, Baccelli2002, Shakkottai2002}. In \cite{Lakshman1997}, the authors compute steady state throughput of TCP Tahoe and TCP Reno under random losses. In \cite{Sharma2004} and \cite{Misra2000}, the authors model transient behaviour of RED routers supporting TCP flows. A mean-field model for TCP is developed in \cite{Baccelli2002} for multiple TCP flows sharing a single bottleneck link which employs RED queue management. Under the same setup of single bottleneck link and multiple TCP flows, the authors in \cite{Shakkottai2002} show that in the many flows regime the data rate evolution and drop rate evolution at the queues can be approximated by a deterministic system.

\subsection{Previous studies of TCP CUBIC and Compound}
The newer variants of TCP have fewer analytical studies. In \cite{Bao2010}, the authors use a Markovian model to compute steady state throughput of a single TCP CUBIC connection in a wireless environment. A mean field model is used for performance analysis of  multiple TCP CUBIC connections going through a single drop-tail bottleneck link in \cite{Belhareth2013}. In \cite{Blanc2009}, the authors compute throughput of a single long-lived Compound TCP under random losses through a Markovian model. There are  deterministic models for computation of average window size of CUBIC and Compound TCP in \cite{Ha2008} and \cite{Tan2006Infocom} respectively. One advantage that these models have over the Markovian models is that they provide a closed-form expression for the average window size of a TCP flow in terms of its RTT and packet error rate. We have investigated these closed form expressions against our Markovian models in \cite{Poojary2013} and have found that the closed form expressions are not as accurate as the Markovian model results. In \cite{Ghosh2014Allerton, Chavan2015, Manjunath2015}, the authors study the performance of TCP Compound using control theoretic techniques and derive stability conditions for TCP Compound. In these papers, it is shown that when multiple TCP flows share a single bottleneck queue, the queue sizes and the link utilization have oscillatory behaviour when the feedback delays (round trip times of the flows) and buffer sizes are large. In \cite{Ghosh2014Allerton}, the authors evaluate the TCP Compound performance as a function of buffer size in the bottleneck queue. In \cite{Chavan2015}, the authors study the performance of TCP Compound with a proportional integral enhanced queue management policy whereas in \cite{Manjunath2015}, RED queue management policy is considered.

\subsection{Our Contribution}
Markov models for TCP CUBIC and TCP Compound do exist in \cite{Bao2010} and \cite{Blanc2009} respectively. However, the Markov model for TCP CUBIC in \cite{Bao2010} assumes a different loss model; the inter packet loss durations are assumed Poisson. In our setup, we assume that packets are lost independently of other packets. This scenario is close to the approach used in \cite{Padhye2000, Kunniyur2003}. We note that the Markov model for TCP Compound in \cite{Blanc2009} is similar to our model. However, our Markov model for TCP Compound, unlike in \cite{Blanc2009}, also includes the queue lengths resulting in better approximation in real world scenario. Also, our Markov chain models for TCP CUBIC and TCP Compound have less computational complexity than those in \cite{Bao2010} and \cite{Blanc2009}.

In this paper, we develop techniques for performance analysis of different TCP flows (using TCP Compound, TCP CUBIC or Reno) over a general network with multiple bottleneck links. There are a number of simulation and experimental evaluations for performance analysis of these TCP variants over different network topologies.  However, to the best of our knowledge, our work is the first theoretical model for the joint performance analysis of these high speed TCP variants over a general network with multiple bottleneck links. Since presence of different types of TCPs affects the throughputs of each other in complicated ways and this is the practical scenario in the current Internet, it is important to study this setup. We validate our model approximations via extensive ns2 simulations.

The organization of our paper is as follows. In Section \ref{sec:system_model}, we describe our system model. We describe a Markovian model to compute average window size of a single TCP CUBIC connection with fixed RTT under random losses in Section \ref{sec:cubic_markovmodel}. In Section \ref{sec:ctcp_markovmodel}, we develop a model for computing the average window size of a single TCP Compound connection with non-negligible queuing under random losses. In Section \ref{sec:multiple_tcp_models}, we describe two techniques viz., M/G/1 approximation and an optimization based approach to compute the steady state average window size and the throughputs for TCP flows (which could be TCP CUBIC, TCP Compound or TCP New Reno) over a multi-hop network. We compare our theoretical results with simulations in Section \ref{sec:simulation_results}. Section \ref{sec:conclusion} concludes our paper.


\section{System Model}
\label{sec:system_model}
\begin{figure}
  \centering
  \includegraphics[scale=0.31]{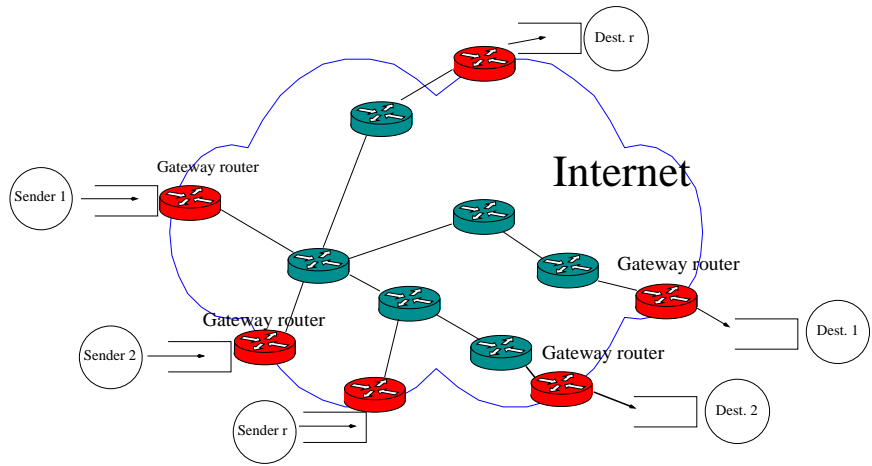}
  \caption{The general network, $(\mathcal{R}, \mathcal{L})$}
  \label{fig:general_network_3}
  \vspace*{-0.3cm}
\end{figure}

Consider a general network of routers as shown in Figure \ref{fig:general_network_3}. A set $\mathcal{R}$ of TCP flows is passing through this network. The TCP flows are carrying long files. A TCP connection may be using TCP New Reno, TCP Compound or TCP CUBIC. We denote the set of links by $\mathcal{L}$. Some of the access links may be wireless. For a flow $r \in \mathcal{R}$, let $\Delta_r$ be the constant round-trip delay (this includes propagation and transmission delays at the links). Let each packet of flow $r$ be lost with probability $p_r$ on its path, independently of others. This indicates that the packet losses in our system are mainly due to transmission errors on the wireless links and neglects the buffer overflows. This is increasingly the scenario in practical networks. We also assume that the packet losses for different flows are independent. 

Let $A$ be the incidence matrix for the network with $A(l,r) = 1$ if packets of flow $r$ go through link $l$. It is possible that besides the TCP packets, the TCP ACKs are also subject to queuing delays (due to congestion on the path from destination to source). Let $B$ be the incidence matrix for the TCP ACK packets, i.e., $B(l,r) = 1$ if TCP ACK packets of flow $r$ go through link $l$. 


We would like to compute the throughput of each TCP connection in this setup. For this, first, in Sections \ref{sec:cubic_markovmodel} and \ref{sec:ctcp_markovmodel}, we derive average window size for a single TCP flow using TCP CUBIC and TCP Compound, respectively, through a single bottleneck link  . Then, using these, in Section \ref{sec:multiple_tcp_models}, we describe techniques to compute the steady state throughputs attained by the different TCP flows in the general network.

\section{A Markov Model for TCP CUBIC}
\label{sec:cubic_markovmodel}
In this section, we develop a model for a single TCP CUBIC connection with constant round trip time (RTT). Any packet received can be in error with probability $p$ independent of other packets. This is a realistic assumption for wireless links and is commonly made for TCP models in the literature (\hspace{-.05mm}\cite{Mathis1997}, \cite{Padhye2000}, \cite{Cardwell2000},  \cite{Lakshman1997},  \cite{Bao2010}, \cite{Blanc2009},  ). We assume that only TCP data packets are lost. The ACKs are typically smaller in size and hence are less likely to be in error. However, the model can be easily extended to the case where ACKs may also be subject to random losses.

We are interested in the performance of long-lived flows (such as file transfers and video streaming) as these are quite common over the current Internet. Hence, we will ignore the slow start phase of window evolution which does not impact throughput for long-lived flows. Also, TCP CUBIC and TCP Compound (also other high speed TCP variants) only differ in the congestion avoidance phase. 

In the congestion avoidance phase, TCP CUBIC uses a non-linear cubic function, ($W_{cubic}$ in equation \eqref{eqn:tcpCUBIC}) for window evolution. However the window sizes obtained by the cubic function can at times be smaller than TCP Reno. In such a scenario, TCP CUBIC uses another function, $W_{reno}(t)$ for window evolution. Assuming that at time $t = 0$, there is a packet loss and there are no further losses in $(0,t)$, the window size, $W(t)$ of TCP CUBIC at time $t$ is given by $W(t) = \max \{ W_{cubic}(t), W_{reno}(t) \}$, where 
\begin{equation}
\begin{split}
W_{cubic}(t) & =  C \Biggl(t - \sqrt[3]{\frac{W_0  \beta}{C}} \Biggr)^3 + W_0, \\
W_{reno}(t) & =   W_{0}  (1 - \beta) + 3  \frac{\beta}{2 - \beta}  \frac{t}{R}, \\
\end{split}
\label{eqn:tcpCUBIC}
\end{equation}
and $W_0$ is the window size at time $t=0$, $C$ is a constant, $\beta$ is the multiplicative drop factor and $R$ is the round trip time of the connection. If there is a packet loss at time $t$, the window size is reduced to $(1 - \beta) W(t)$. 

Although in principle the equations in \eqref{eqn:tcpCUBIC} are evolving continuously in time, because of constant RTT, we can assume that the TCP source updates its window size at the beginning of each RTT and transmits that many number of packets. This is usually assumed in TCP studies and will be validated via simulations. The evolution of TCP CUBIC, as given by \eqref{eqn:tcpCUBIC} can be modeled using a simple Markov chain. Let $W_n \in \{1, 2, \cdots \}$ denote the window size at the beginning of the $(n+1)^{st}$ RTT. Let $W_n^{\prime}$ denote the window size at the last packet loss epoch prior to the end of the $n^{th}$ RTT and $T_n$ denote the number of RTT epochs elapsed between the last packet loss and the $n^{th}$ RTT. If there is no loss between the $n^{th}$ and the $(n+1)^{st}$ RTT, $W_{n+1}^{\prime} = W_n^{\prime}$ and $T_{n+1} = T_{n} + 1$. If there is a loss we use \eqref{eqn:tcpCUBIC} and set
\begin{equation}
\label{eqn:tcpCUBIC_MC_no_loss}
W_n = \max \{ C \Bigl(R T_n - \sqrt[3]{\frac{W_n^{\prime}  \beta}{C}} \Bigr)^3 + W_n^{\prime} ,  W_n^{\prime}  (1 - \beta) + 3  \frac{\beta}{2 - \beta}  T_n \},
\end{equation}
$W_{n+1}^{\prime} = W_n$ and $T_{n+1} = 0$. The probability of no loss between the $n^{th}$ and the $(n+1)^{st}$ RTT is given by $(1-p)^{W_n}$. We illustrate the $\{W_n, W_n^{\prime}, T_n\}$ processes in Figure \ref{fig:notation_2}. In the Figure, we assume that there is a packet loss just before the first RTT. On packet loss, the value of $W_n^{\prime}$ is updated and $T_n$ is set to $0$. For example, since there is a loss in RTT $3$ in Figure \ref{fig:notation_2}, at the end of the $3^{rd}$ RTT, $W_3^{\prime}$ is set to $W_2$ and $T_3$ is set to $0$. Also,  $W_3$ is set to $W_2(1-\beta)$.

\begin{figure}
  \centering
  \includegraphics[scale=0.35]{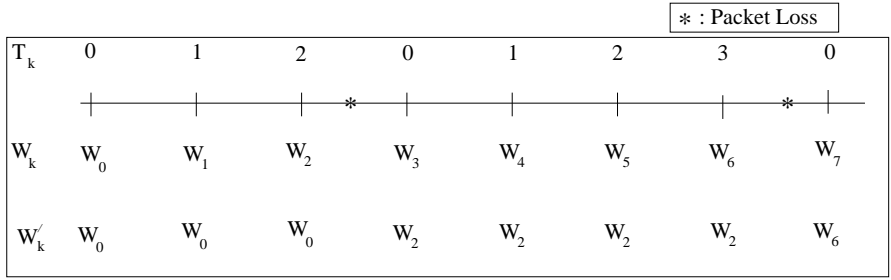}
  \caption{TCP CUBIC: Illustrating $\{W_n, W_n^{\prime}, T_n\}$ processes}
  \label{fig:notation_2}
\end{figure}

The TCP window size is usually restricted by the buffer size available at the receiver. Let the window size $W_n \leq W_{max} < \infty$. Therefore $W_n^{\prime} \leq W_{max}$. Let $T_{max}$ be the maximum time (in multiples of RTT) taken for the TCP CUBIC window size, $W(t)$  to hit $W_{max}$ starting from any initial window size $W_0 \in \{1, 2, \cdots W_{max} \}$. We note that if $T_n$ exceeds $T_{max}$, then $W_n$ as computed by \eqref{eqn:tcpCUBIC_MC_no_loss} exceeds $W_{max}$ and in this case we set $W_n$ to $W_{max}$. Therefore, we can restrict $T_n$ to be in $\{1, 2, \cdots, T_{max}\}$. Thus the process $\{W_n^{\prime}, T_n\}$ forms a finite state discrete time Markov chain. From any state in the state space, a sequence of consecutive packet drops would cause the Markov chain to hit $(1,0)$. Therefore, the state $(1,0)$ can be reached with positive probability from any state in the state space and hence is recurrent. All states that can be reached from $(1,0)$ are recurrent and the remaining states are transient. Also, there is a self loop from state $(1,0)$ to itself. Therefore the process $\{W_n^{\prime}, T_n\}$ is aperiodic with a single positive recurrent communicating class. Hence it has a unique stationary distribution which we denote by $\pi(w,d)$. Also, starting from any initial state the chain converges exponentially to the stationary distribution in total variation.

Let $\mathbb{E}[W]$ denote the mean window size under stationarity. For TCP CUBIC, it can be computed as
\begin{equation}
\label{eqn:EW_CUBIC}
\mathbb{E}[W] = \sum_{1 \leq w \leq W_{max}, d \in \mathcal{D} } W(w,d) \pi(w,d),
\end{equation}
where $\mathcal{D} \subset \{0, 1, \cdots T_{max}\}$ and $W(w,d)$ is given by
\begin{equation}
\label{eqn:Wof_w_and_d}
W(w,d) = \max \{ C \Bigl(R d - \sqrt[3]{\frac{w  \beta}{C}} \Bigr)^3 + w ,  w  (1 - \beta) + 3  \frac{\beta}{2 - \beta}  d \}.
\end{equation}
From \eqref{eqn:tcpCUBIC_MC_no_loss}, we see that the transitions of the process $\{W_n^{\prime}, T_n\}$ (both the next state and the transition probabilities) depend on the RTT of the connection. Therefore the stationary distribution, $\pi(w,d)$ and the mean window size of TCP CUBIC is a function of the RTT of the connection. We denote the relationship between RTT, $R$, packet error rate, $p$ and the mean window size, $\mathbb{E}[W]$ for TCP CUBIC as 
\begin{equation}
\label{eqn:EW_TCP_CUBIC_MC}
\mathbb{E}[W] = f_p(R),
\end{equation}
which is numerically computed using \eqref{eqn:EW_CUBIC} (see Figures \ref{fig:Effect_RTT_TCPCUBIC_1} and \ref{fig:Effect_RTT_TCPCUBIC_2} below). We note that \eqref{eqn:EW_TCP_CUBIC_MC} is not a closed form expression but a numerical evaluation of \eqref{eqn:EW_CUBIC}. We use \eqref{eqn:EW_TCP_CUBIC_MC} in Section \ref{sec:multiple_tcp_models} where we study TCP CUBIC in a multi-hop network where the RTT may not be constant due to queuing in the network. In that case, we use an approximation and replace $R$ by the mean RTT, $\mathbb{E}[R]$, of the flow.

We do not consider the case of non-negligible queuing for TCP CUBIC. However, in Section \ref{sec:multiple_tcp_models}, where we consider TCP connections with non-negligible queuing, we approximate the TCP CUBIC mean window size $\mathbb{E}[W] \approx f_p(\mathbb{E}[R])$, where $\mathbb{E}[R]$ is the mean RTT of the connection.

\subsection{Simulation Results}
In Figures \ref{fig:Effect_WMax_TCPCUBIC_01_2},  \ref{fig:Effect_WMax_TCPCUBIC_02} and \ref{fig:Effect_WMax_TCPCUBIC_03}, we plot the mean window size $\mathbb{E}[W]$ as a function of $W_{max}$ and compare it to ns2 simulations. In our simulations and model, we set $C = 0.4, \beta = 0.3$, which are the values used by the current version of TCP CUBIC \cite{tcp_cubic_code}. The packet sizes are set to $1050$ bytes which is the default value in ns2. The link speeds are set to $1$ Gbps. Each packet can be dropped independently of other packets with probability $p$.  The loss in the network is modelled as random. Hence, in all the simulations, we set the link buffer size to be greater than $W_{max}$ so that there are no buffer drops. In Figure \ref{fig:Effect_WMax_TCPCUBIC_01_2}, we plot results for packet error rates, $0.01$, $0.008$, $0.005$ and $0.003$ and RTT is set to $0.2$ sec (bandwidth delay product $= 23810$ packets). In Figure \ref{fig:Effect_WMax_TCPCUBIC_02}, we plot results for packet error rates, $0.001$, $0.0003$ and $0.0001$ and RTT is set to $0.02$ sec (bandwidth delay product $= 2381$ packets). In Figure \ref{fig:Effect_WMax_TCPCUBIC_03}, we plot results for packet error rates, $5 \times 10^{-5}$, $3 \times 10^{-5}$ and $1 \times 10^{-5}$ and RTT is set to $0.2$ sec (bandwidth delay product $= 23810$ packets). The theoretical and ns2 results differ by $ < 8.5 \%$. The mean window size $\mathbb{E}[W]$ increases monotonically with $W_{max}$. However for large values of $W_{max}$, change in $W_{max}$ has negligible effect on $\mathbb{E}[W]$. This suggests that for large values of $W_{max}$, $\mathbb{E}[W]$ is not affected by $W_{max}$. In the rest of the section, we will be working in the regime where $W_{max}$ has no effect on $\mathbb{E}[W]$.

\begin{figure}
  \centering
  \includegraphics[scale=0.25, trim = 120 5 120 5, clip=true]{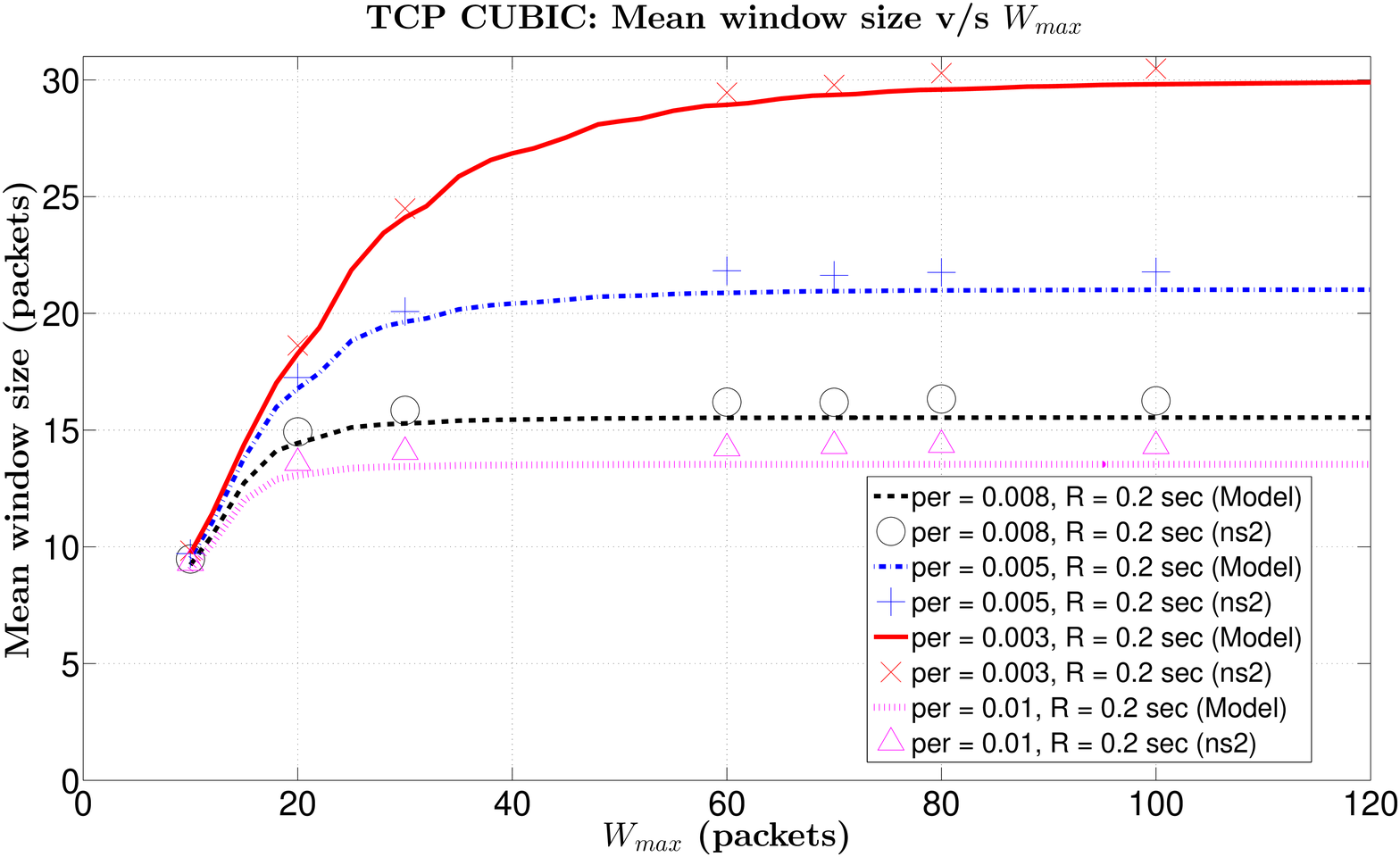}
  \caption{TCP CUBIC: Effect of $W_{max}$ on $\mathbb{E}[W]$.}
  \label{fig:Effect_WMax_TCPCUBIC_01_2}
  \vspace*{-0.3cm}
\end{figure}

\begin{figure}
  \centering
  \includegraphics[scale=0.25, trim = 120 5 120 5, clip=true]{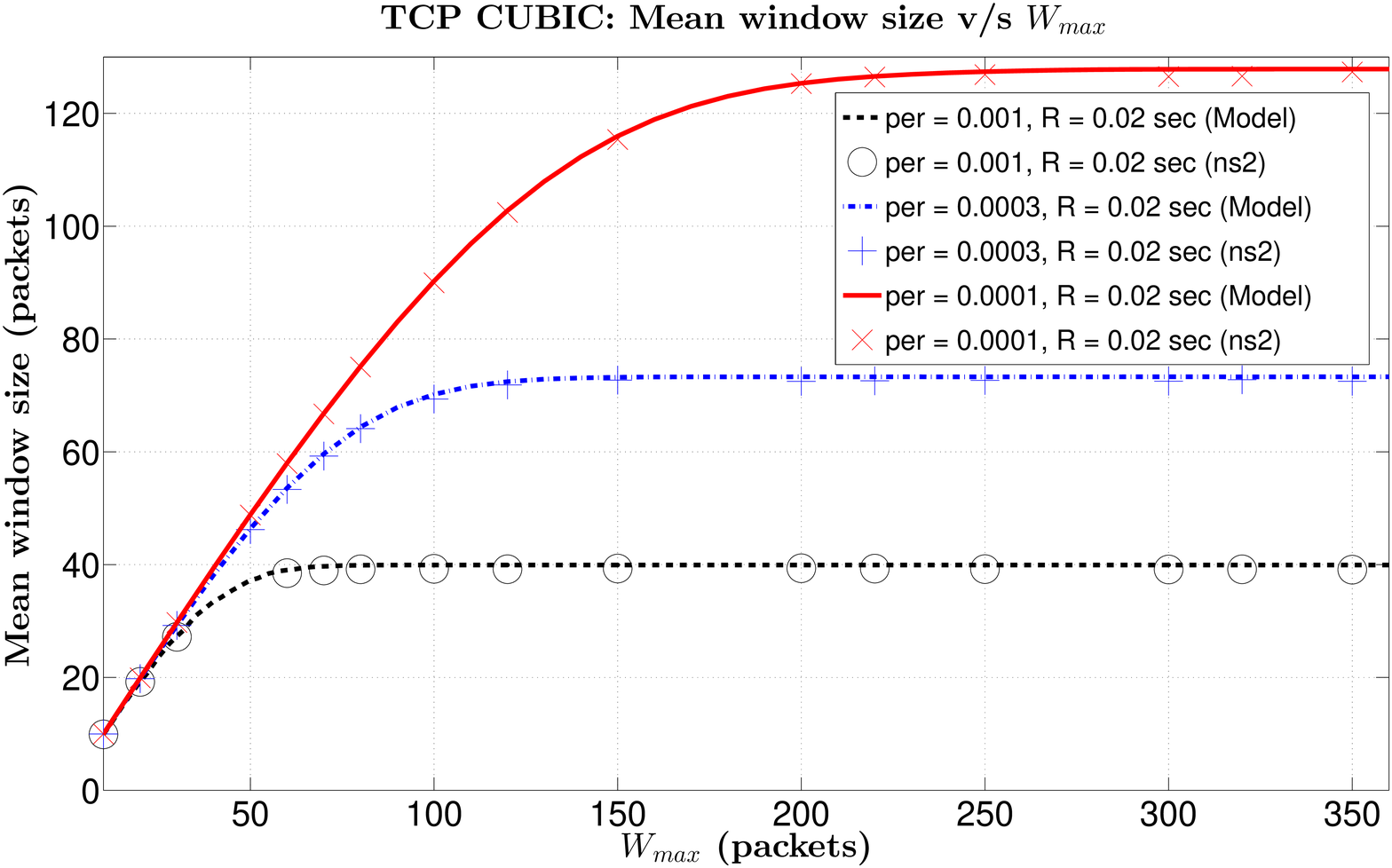}
  \caption{TCP CUBIC: Effect of $W_{max}$ on $\mathbb{E}[W]$.}
  \label{fig:Effect_WMax_TCPCUBIC_02}
  \vspace*{-0.3cm}
\end{figure}

\begin{figure}
  \centering
  \includegraphics[scale=0.25, trim = 80 5 120 5, clip=true]{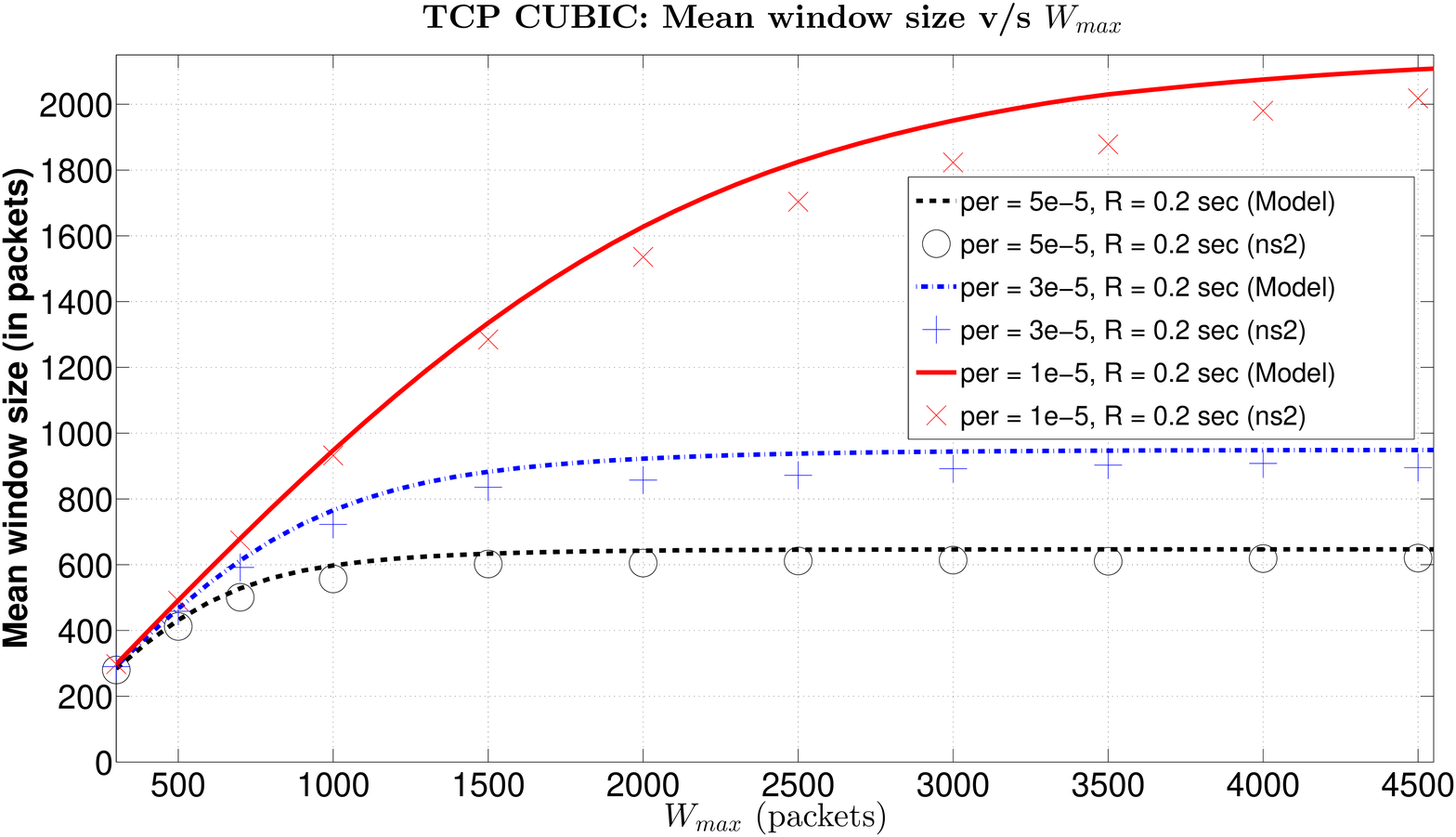}
  \caption{TCP CUBIC: Effect of $W_{max}$ on $\mathbb{E}[W]$.}
  \label{fig:Effect_WMax_TCPCUBIC_03}
  \vspace*{-0.3cm}
\end{figure}


In Figures \ref{fig:Effect_RTT_TCPCUBIC_1} and \ref{fig:Effect_RTT_TCPCUBIC_2}, we plot our results for mean window size for TCP CUBIC as a function of the round trip time (RTT) and compare them to ns2 simulations. The packet sizes are set to  $1050$ bytes. The link speeds are set to $1$ Gbps. The bandwidth-delay product in these simulations range from $1190$ to $59523$ packets. The simulation and theoretical results differ by $ < 8 \%$. We see in these figures that, unlike TCP Reno, the mean window size of TCP CUBIC is a function of the RTT along with the packet error rate. Traditional TCP viz., AIMD TCP (e.g., TCP Reno, TCP New Reno) is known to be unfair to TCP connections with longer RTT. Since the mean window size of CUBIC grows with RTT, TCP CUBIC has better RTT-fairness than AIMD TCP. This feature also makes TCP CUBIC efficient in networks with larger RTT.

\begin{figure}
  \centering
  \includegraphics[scale=0.25, trim = 80 5 120 5, clip=true]{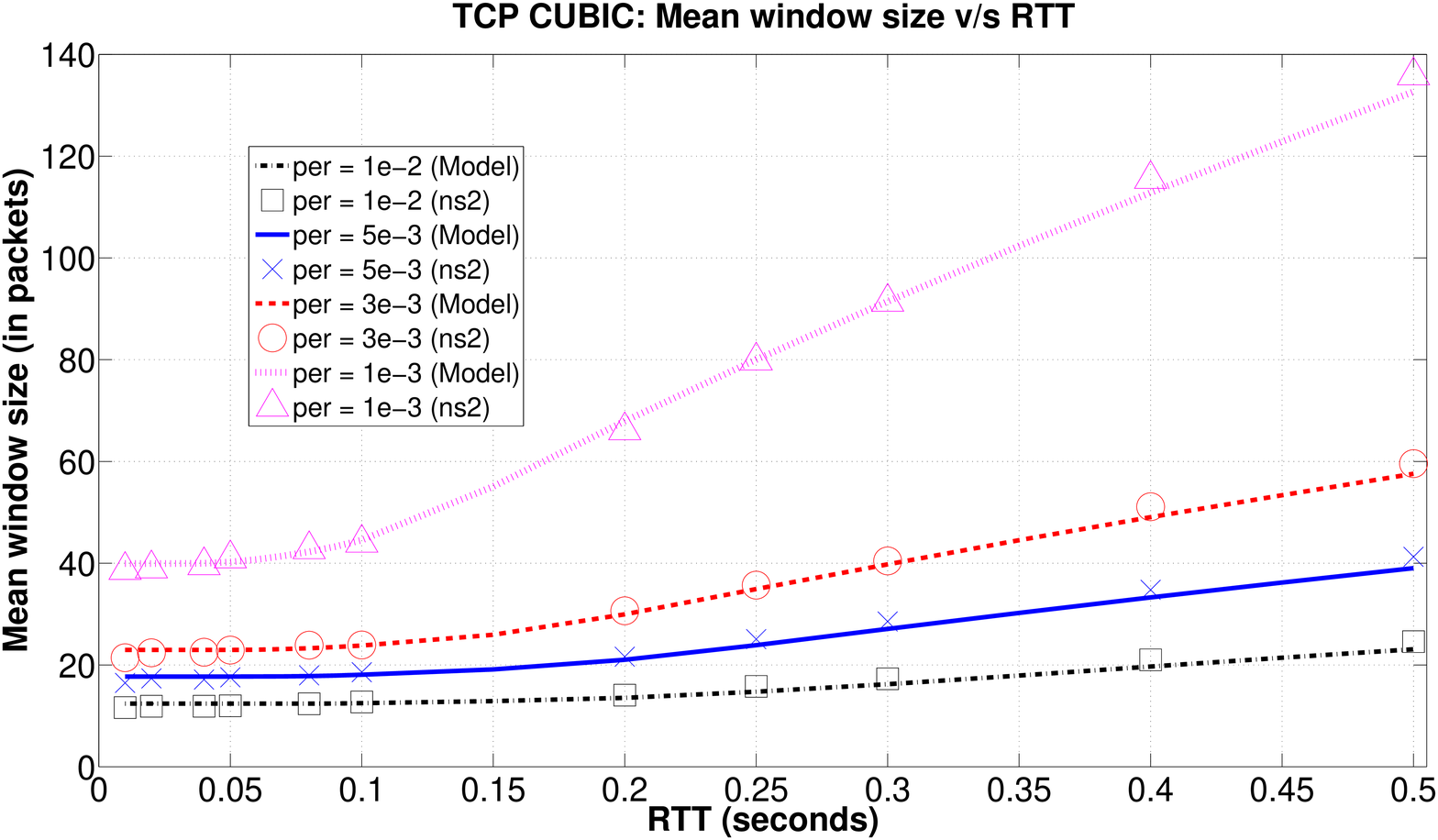}
  \caption{TCP CUBIC: Effect of RTT on $\mathbb{E}[W]$.}
  \label{fig:Effect_RTT_TCPCUBIC_1}
  \vspace*{-0.3cm}
\end{figure}

\begin{figure}
  \centering
  \includegraphics[scale=0.25, trim = 80 5 120 5, clip=true]{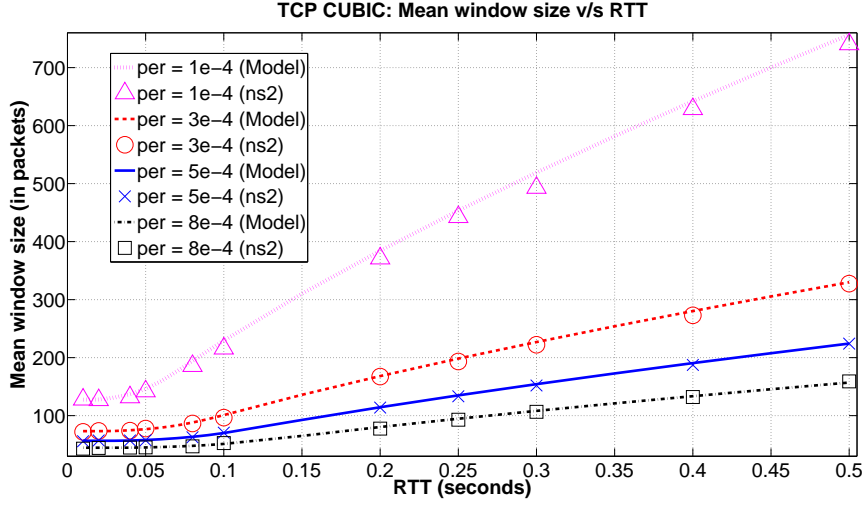}
  \caption{TCP CUBIC: Effect of RTT on $\mathbb{E}[W]$.}
  \label{fig:Effect_RTT_TCPCUBIC_2}
  \vspace*{-0.3cm}
\end{figure}

\section{A Markov Model for TCP Compound}
\label{sec:ctcp_markovmodel}
In this section, we develop Markov models for TCP Compound connections with negligible and non-negligible queuing. As before, we have a single TCP connection. We assume that any packet received can be in error with probability $p$ independent of other packets. Also, we assume that ACKs are not lost. 

TCP Compound is a delay-based congestion control algorithm. When competing with flows using loss-based congestion control algorithms, flows using delay based congestion control algorithms get less than their fair share of capacity. To address this drawback, TCP Compound window has a loss based component besides the delay based component. The delay based component increases rapidly when there is no congestion in the network and decreases when there is congestion. The loss-based component ensures that when there is congestion in the network, TCP Compound flows behave like TCP-Reno, thus getting their fair share of the network capacity. We denote the window size of TCP Compound at the end of the $n^{th}$ RTT by $W_n$ and the delay based and loss based components by $D_n$ and $L_n$ respectively. The window evolution of TCP Compound is given by 

\begin{equation}
\label{eqn:Dn}
D_{n+1} = 
\begin{cases}
D_{n} + (\alpha (W_{n})^{k} - 1)^{+}, \\ \hspace*{0.5cm}\text{ if no loss during the RTT and } Q_{n+1} < \gamma; \\
(D_{n} - \zeta Q_{n+1})^{+}, \\ \hspace*{0.5cm} \text{ if no loss during the RTT and } Q_{n+1} \geq \gamma; \\
\frac{D_{n}}{2}, \text{ if loss is detected};
\end{cases}
\end{equation}

\begin{equation}
\label{eqn:Ln}
L_{n+1} = 
\begin{cases}
L_{n} + 1, \text{ if no loss}; \\
\frac{L_{n}}{2}, \text{ if a loss is detected;}
\end{cases}
\end{equation}

\begin{equation}
\label{eqn:CTCP_Wn}
W_{n+1} = D_{n+1} + L_{n+1};
\end{equation}
where $\alpha$ and $k$ are constant parameters and $\gamma$ is a queuing threshold. We see that the loss-based component $L_n$ has behaviour similar to TCP Reno. The delay based component $D_n$ increases aggressively when there is no queuing but decreases when queuing increases beyond threshold $\gamma$. The variable $Q_{n+1}$ is the estimate of queuing in the network at the end of the $(n+1)^{st}$ RTT. 

\subsection{Markov Model for Negligible Queuing}
\label{subsec:ctcp_markov_noq}
When there is no queuing, the window evolution given in equations \eqref{eqn:Dn}, \eqref{eqn:Ln} and \eqref{eqn:CTCP_Wn} simplifies to
\begin{equation}
\label{eqn:tcpCTCP_Wn}
W_{n+1} = 
\begin{cases}
W_n + 1 + (\alpha W_n^k - 1)^+, \text{ if there is no loss}; \\
\frac{W_n}{2}, \text{ if there is loss.}
\end{cases}
\end{equation}
The probability of no loss between the $n^{th}$ and the $(n+1)^{st}$ RTT is given by $(1-p)^{W_n}$, whereas the probability of loss is given by $1 - (1-p)^{W_n}$. Thus, $\{W_n\}$ is a finite state discrete time Markov chain with $W_n \leq W_{max}$.

The state $1$ can be reached with positive probability after a sequence of consecutive drops from any state in the state space and hence is recurrent. All states that can be reached from $1$ are recurrent and the remaining states are transient. Also, the Markov chain is aperiodic as the state $1$ has a self-loop with self loop transition probability $p$. Hence, the Markov chain has a unique stationary distribution which can be used to compute the average window size,  $\mathbb{E}[W]$. Also, starting from any initial state the chain converges exponentially to the stationary distribution in total variation. 

From \eqref{eqn:tcpCTCP_Wn}, we see that, when queuing is negligible, unlike TCP CUBIC, for TCP Compound the next state and the transition probabilities are independent of the RTT, $R$ of the connection. Therefore, in this case, the TCP Compound average window size is independent of the RTT of the connection.

\subsection{Markov Model for Non-negligible Queuing}
\label{subsec:ctcp_markov_q}
We now consider the case when the queuing is not necessarily negligible. In this case, we have to consider both the delay and the loss based components of the window size $W_n$. Given $D_n$, $L_n$ and $Q_{n+1}$, the components $D_{n+1}$ and $L_{n+1}$ can be computed from equations \eqref{eqn:Dn} and \eqref{eqn:Ln} respectively. We approximate the queue size, $Q_{n+1}$ at the end of the $(n + 1)^{st}$ RTT  by $Q_{n+1} \approx (W_n - \mu \Delta)^+$ where $\mu$ is the bottleneck link capacity in packets/sec and $\Delta$ is the constant propagation delay. This approximation is based on the assumption that,  when the window size is $W$, the instantaneous rate at which packets are sent is $\min \{ \frac{W}{\Delta}, \mu \}$. Thus the queue size is $(W -  \min \{ \frac{W}{\Delta}, \mu \} \Delta)$ $= (W - \mu \Delta)^+$. We validate this approximation by simulation results. The process $\{(D_n, L_n)\}$ forms a Markov chain. The transitions are given by equations \eqref{eqn:Dn} and \eqref{eqn:Ln}. The probability of no loss between the $n^{th}$ and the $(n+1)^{st}$ RTT is given by $(1-p)^{W_n}$, whereas the probability of loss is given by $1 - (1-p)^{W_n}$. We assume that the window size is upper bounded by $W_{max}$. Thus the Markov chain has finite state space.

From any state in the state space, a sequence of consecutive packet drops would cause the Markov chain to hit $(0,1)$. Therefore, the state $(0, 1)$ can be reached from any state in the state space with positive probability. Hence $(0, 1)$ is positive recurrent and all states that can be reached from $(0, 1)$ are positive recurrent and the remaining states are transient. Also the state $(0,1)$ has a self loop with self-loop transition probability $p$. Therefore the Markov chain $\{(D_n, L_n)\}$ is aperiodic and has a unique stationary distribution which we denote by $\pi(d,l)$. Also, starting from any initial state the chain converges exponentially to the stationary distribution in total variation.

Let us denote the stationary value of the $\{W_n\}$ process, where $W_n = D_n + L_n$, by $\overline{W}$. Suppose $R_n$ denotes the RTT of the packets in the window transmitted at the end of $n^{th}$ RTT. We can approximate $R_n$ as $R_n \approx \max \{\Delta, \frac{W_n}{\mu} \}$. This approximation is quite commonly made \cite{Bonald1999}, \cite{Blanc2009}. Let $\overline{R}$ be a random variable with the stationary distribution of $\{R_n\}$. Let us denote the stationary window size for the continuous time window size process, $W(t)$  by $\mathcal{W}$. Let $\hat{R_k}$ be the RTT for the $k^{th}$ packet. We denote the stationary value for $\hat{R_k}$ by $\hat{R}$. The time average window size and the packet average RTT for the connection can be computed using Palm calculus \cite{Asmussen} as 
\begin{equation}
\label{eqn:palm1}
\mathbb{E}[\mathcal{W}] = \frac{\mathbb{E}[\overline{W} \hspace*{0.5mm} \overline{R}]}{\mathbb{E}[\overline{R}]}, \mbox{ }\mbox{ } \mathbb{E}[\hat{R}] = \frac{\mathbb{E}[\overline{W} \hspace*{0.5mm}  \overline{R}]}{\mathbb{E}[\overline{W}]},
\end{equation}
where the averages $\mathbb{E}[\overline{W} \hspace*{0.5mm}  \overline{R}]$, $\mathbb{E}[\overline{R}]$ and $\mathbb{E}[\overline{W}]$ are computed using $\pi(d,l)$. 


\subsection{Simulation Results}
The results obtained using the Markov model for Compound TCP with negligible queuing are compared with ns2 simulations in Figures \ref{fig:Effect_WMax_CTCP_01}, \ref{fig:Effect_WMax_CTCP_02} and \ref{fig:Effect_WMax_CTCP_03}. The packet sizes are set to $1050$ bytes. The link speeds are set to $1$ Gbps so that there is negligible queuing. Each packet can be dropped independently of other packets with probability $p$.  We model the losses in the network as random. Hence, in all the simulations, we set the link buffer size to be greater than $W_{max}$, so that there are no buffer drops. In Figure \ref{fig:Effect_WMax_CTCP_01}, we plot results for packet error rates, $0.008$, $0.005$ and $0.003$. In Figure \ref{fig:Effect_WMax_CTCP_02}, we plot results for packet error rates, $0.001$, $0.0008$ and $0.0003$. In Figure \ref{fig:Effect_WMax_CTCP_03}, we plot results for packet error rates,  $5 \times 10^{-5}$, $3 \times 10^{-5}$ and $1 \times 10^{-5}$. We use two different values for the propagation delay, viz., $0.2$ sec (bandwidth delay product $= 23810$ packets) and $0.02$ sec (bandwidth delay product $= 2381$ packets) for Figures \ref{fig:Effect_WMax_CTCP_01} and \ref{fig:Effect_WMax_CTCP_02} and $0.2$ sec (bandwidth delay product $= 23810$ packets) and $0.1$ sec (bandwidth delay product $= 1190$ packets) for Figure \ref{fig:Effect_WMax_CTCP_03}. In our model, the average window size does not depend on the RTT of the flow. In Figures \ref{fig:Effect_WMax_CTCP_01}, \ref{fig:Effect_WMax_CTCP_02} and \ref{fig:Effect_WMax_CTCP_03} we observe that for the ns2 simulations, the average window size of TCP Compound flows with same packet error rate but different RTT are close to each other. Thus, we see that when the RTT is constant, i.e., queuing is negligible, the average window size of TCP Compound flow does not depend  on its RTT. Also, as in the case of TCP CUBIC, for large values of $W_{max}$, there is no change in the average window size as $W_{max}$ changes. In the rest of the section, we will be working in this regime, i.e., we choose $W_{max}$ large enough so that there is no effect of $W_{max}$ on the average window size.  

\begin{figure}
  \centering
  \includegraphics[scale=0.25, trim = 120 5 120 5, clip=true]{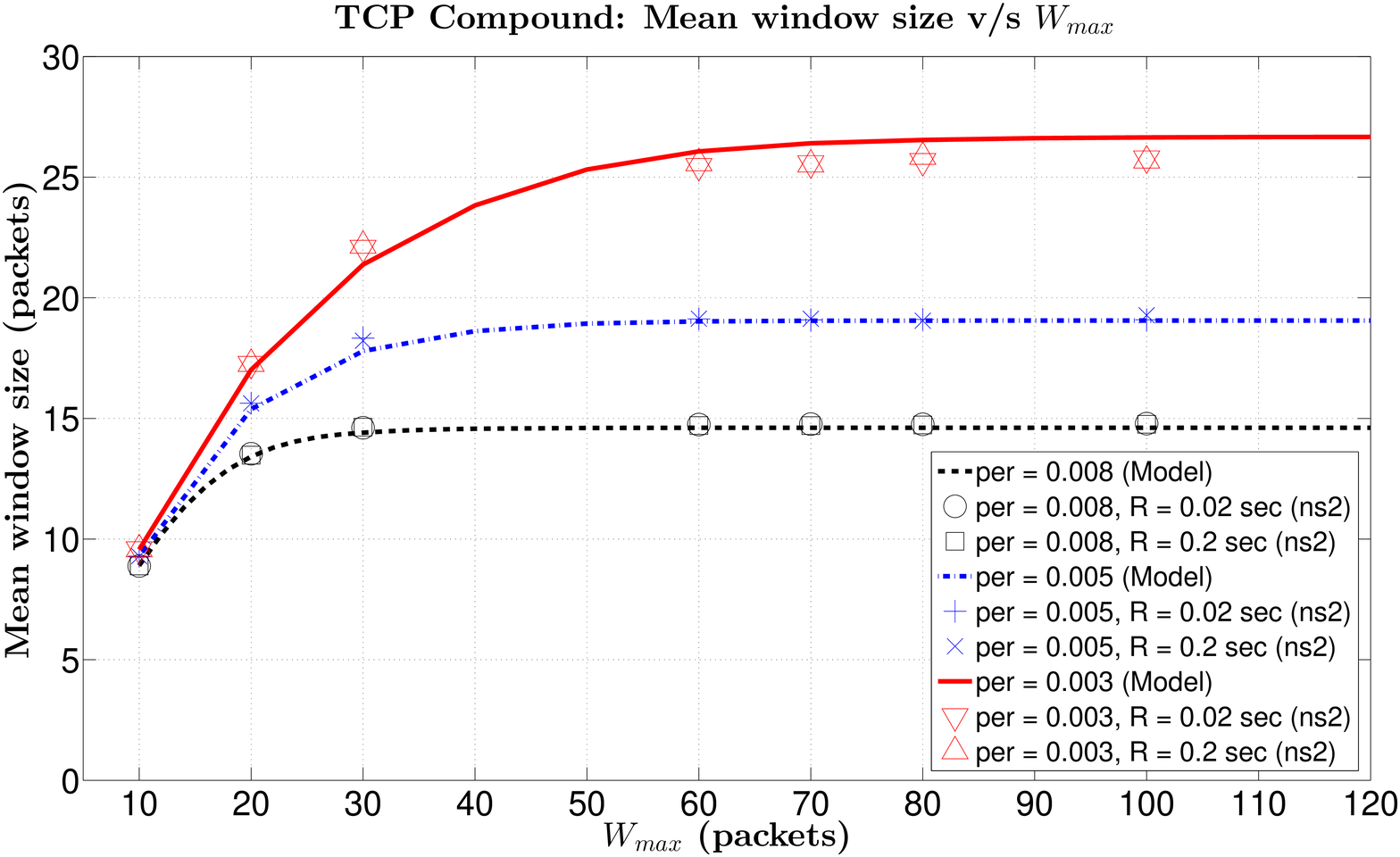}
  \caption{TCP Compound: Effect of $W_{max}$ on $\mathbb{E}[W]$, negligible queuing.}
  \label{fig:Effect_WMax_CTCP_01}
\end{figure}

\begin{figure}
  \centering
  \includegraphics[scale=0.25, trim = 120 5 120 5, clip=true]{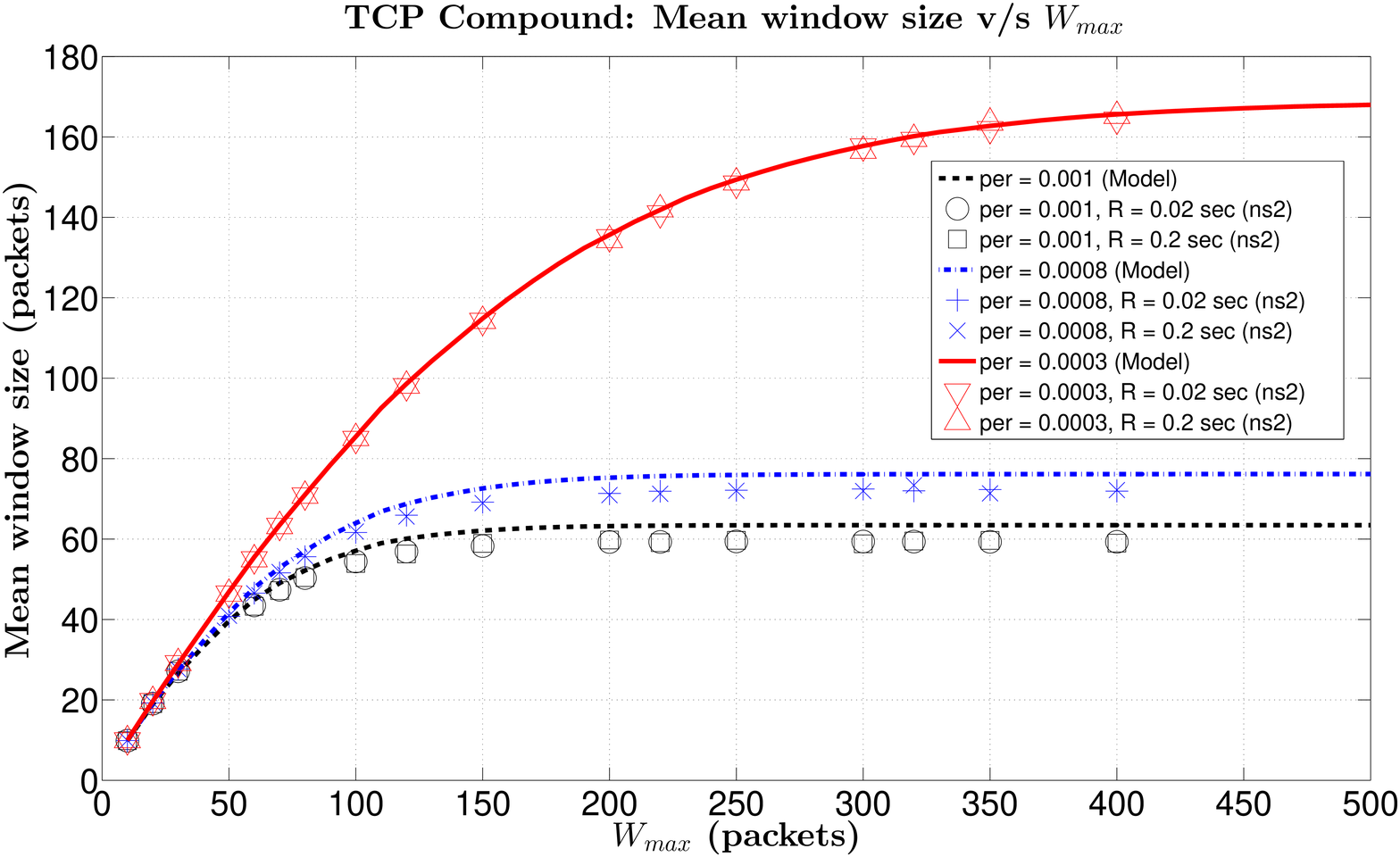}
  \caption{TCP Compound: Effect of $W_{max}$ on $\mathbb{E}[W]$, negligible queuing.}
  \label{fig:Effect_WMax_CTCP_02}
\end{figure}

\begin{figure}
  \centering
  \includegraphics[scale=0.25, trim = 120 5 120 5, clip=true]{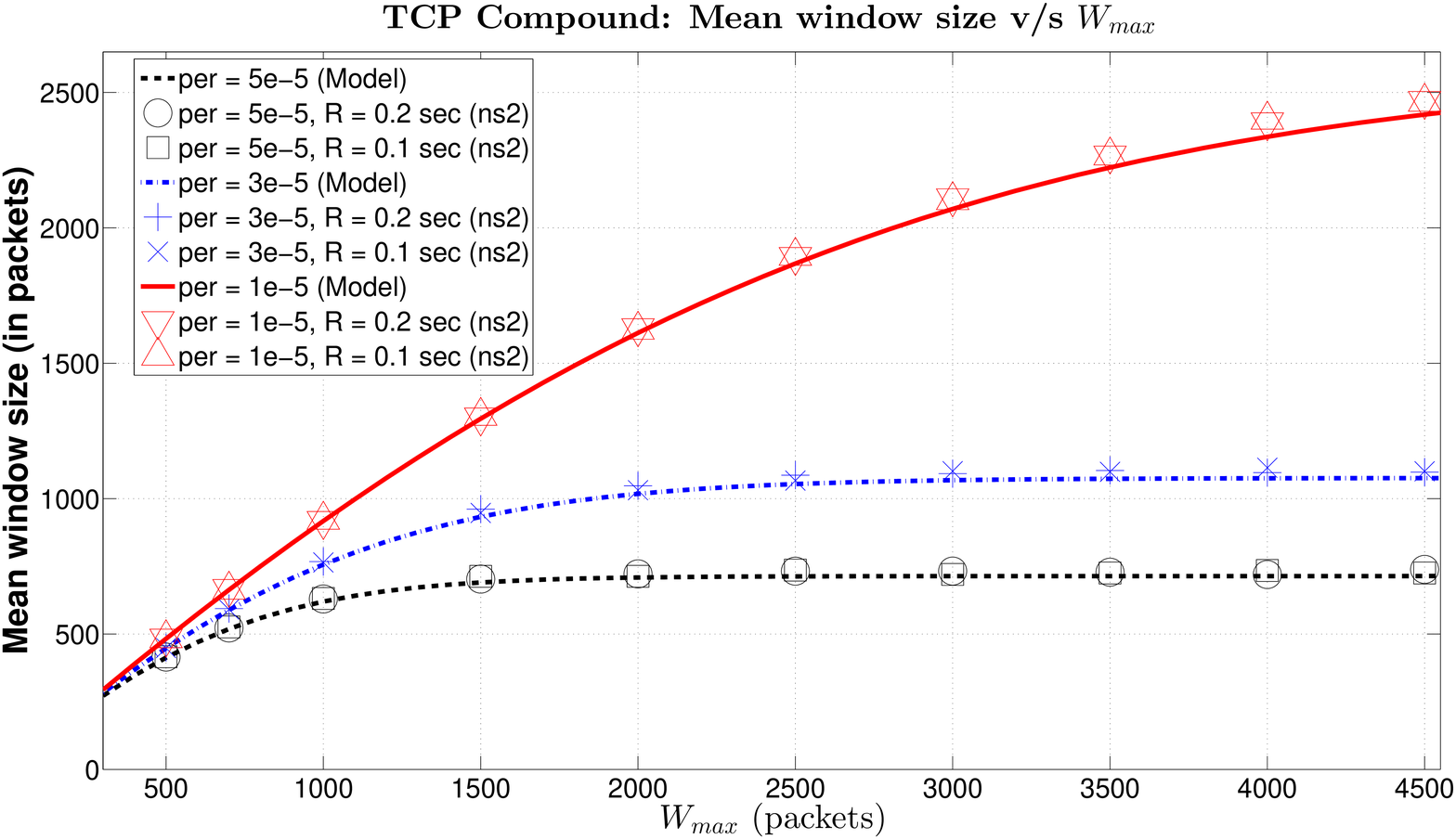}
  \caption{TCP Compound: Effect of $W_{max}$ on $\mathbb{E}[W]$, negligible queuing.}
  \label{fig:Effect_WMax_CTCP_03}
\end{figure}


In Figures \ref{fig:Effect_RTT_CTCP_1} and \ref{fig:Effect_RTT_CTCP2}, we plot the mean window size for TCP Compound with bottleneck link speed $C = 1$ Mbps. In Figures \ref{fig:CTCP_link_utilization_01} and \ref{fig:CTCP_link_utilization}, we plot the normalized link utilization in this case. The packet sizes are set to $1050$ bytes. The bandwidth delay product in these simulations range from $1$ packet to $60$ packets. In Figure \ref{fig:Effect_RTT_CTCP_3} and \ref{fig:CTCP_link_utilization_03}, we plot the mean window size and normalized link utilization with bottleneck link speed $C = 10$ Mbps. The bandwidth delay product in these simulations range from $12$ packet to $595$ packets. The simulation and model results differ by $< 10.5\%$.

In Figures \ref{fig:Effect_RTT_CTCP2} and \ref{fig:Effect_RTT_CTCP_3}, we see that as the round trip propagation delay, $\Delta$ increases, the average window size of TCP Compound increases. However the change in average window size is not due to change in $\Delta$, but due to queuing. Due to the delay based component of Compound TCP, flows which encounter more queuing have smaller average window sizes. For a fixed bottleneck link capacity, flows with larger propagation delays have smaller throughput and hence smaller queues at the bottleneck link (by Little's law). Therefore when there is non-negligible queuing, flows with larger propagation delays have larger average window sizes. When window sizes are small, the delay based component has negligible contribution to the window size and TCP Compound behaves like Reno. This behaviour can be seen in Figure \ref{fig:Effect_RTT_CTCP_1} for packet error rates of $0.01$, $0.008$ and $0.005$ where there is not much change in mean window size as $\Delta$ changes. In Figures \ref{fig:CTCP_link_utilization_01}, \ref{fig:CTCP_link_utilization} and  \ref{fig:CTCP_link_utilization_03}, we see that as propagation delay, $\Delta$ increases, link utilization decreases. For larger $\Delta$, the random packet errors become a bottleneck and adversely affect link utilization.

The average window size for TCP Compound is a function of packet error rate and the queuing that the flow causes in the network. We use the following approximation for TCP Compound average window size
\begin{equation}
\label{eqn:EW_TCP_CTCP_MC}
\mathbb{E}[W] \approx g_p(\mathbb{E}[Q]),
\end{equation}
where $p$ is the packet error rate for the flow and $\mathbb{E}[Q]$ is the average number of packets in the queue. The equation \eqref{eqn:EW_TCP_CTCP_MC} is not a closed form expression but is obtained numerically.

\begin{figure}
  \centering
  \includegraphics[scale=0.25, trim = 120 5 120 5, clip=true]{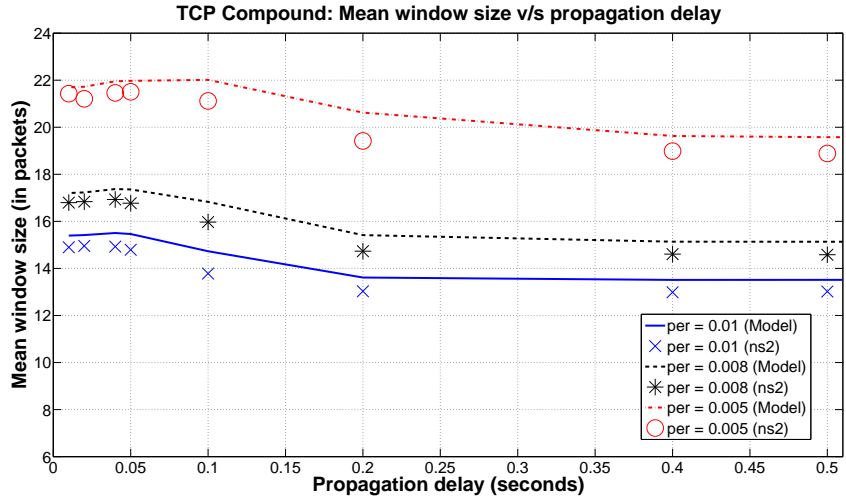}
  \caption{TCP Compound: Effect of change in BDP on $\mathbb{E}[W]$.}
  \label{fig:Effect_RTT_CTCP_1}
\end{figure}

\begin{figure}
  \centering
  \includegraphics[scale=0.25, trim = 120 5 120 5, clip=true]{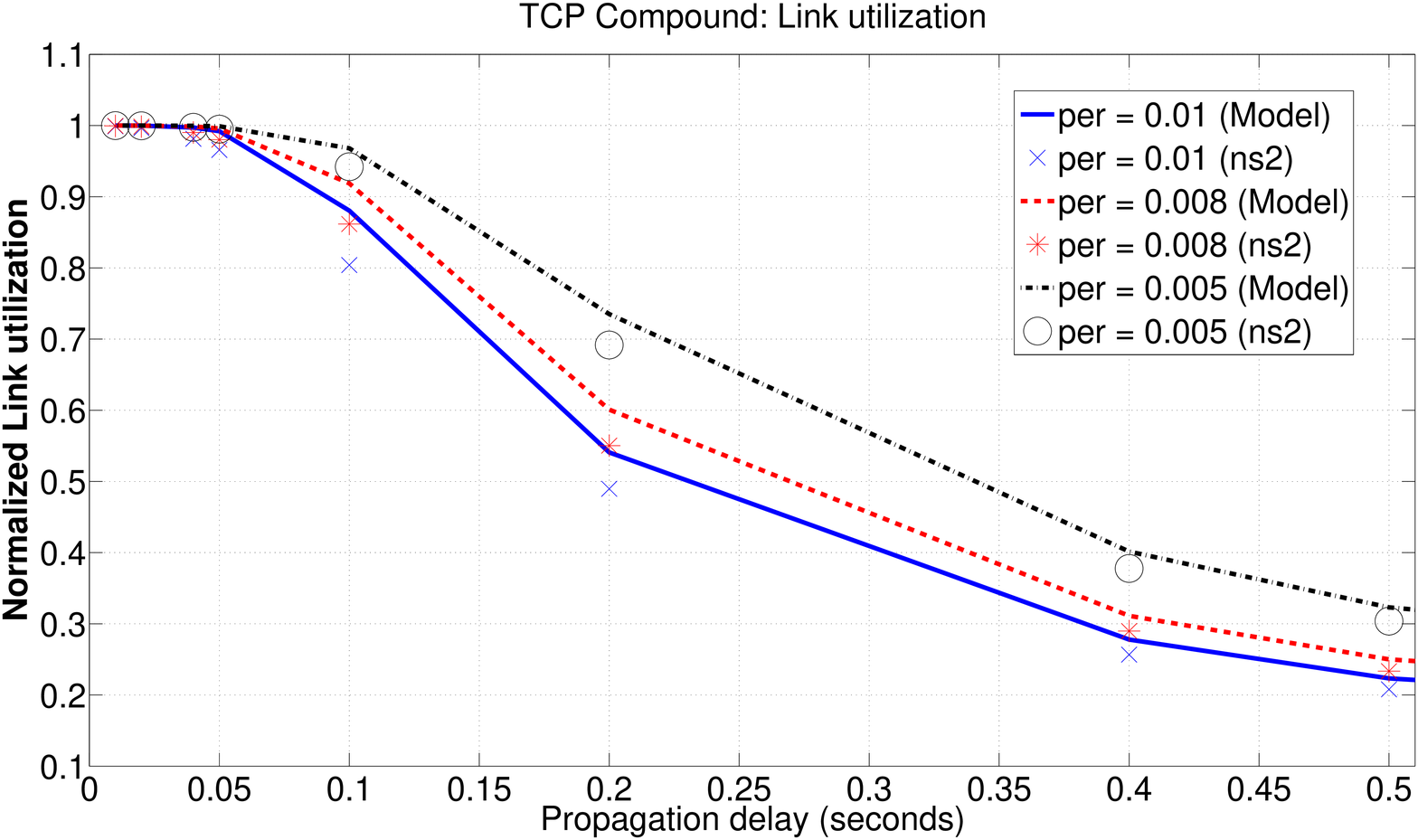}
  \caption{TCP Compound: link utilization.}
  \label{fig:CTCP_link_utilization_01}
\end{figure}

\begin{figure}
  \centering
  \includegraphics[scale=0.25, trim = 120 5 120 5, clip=true]{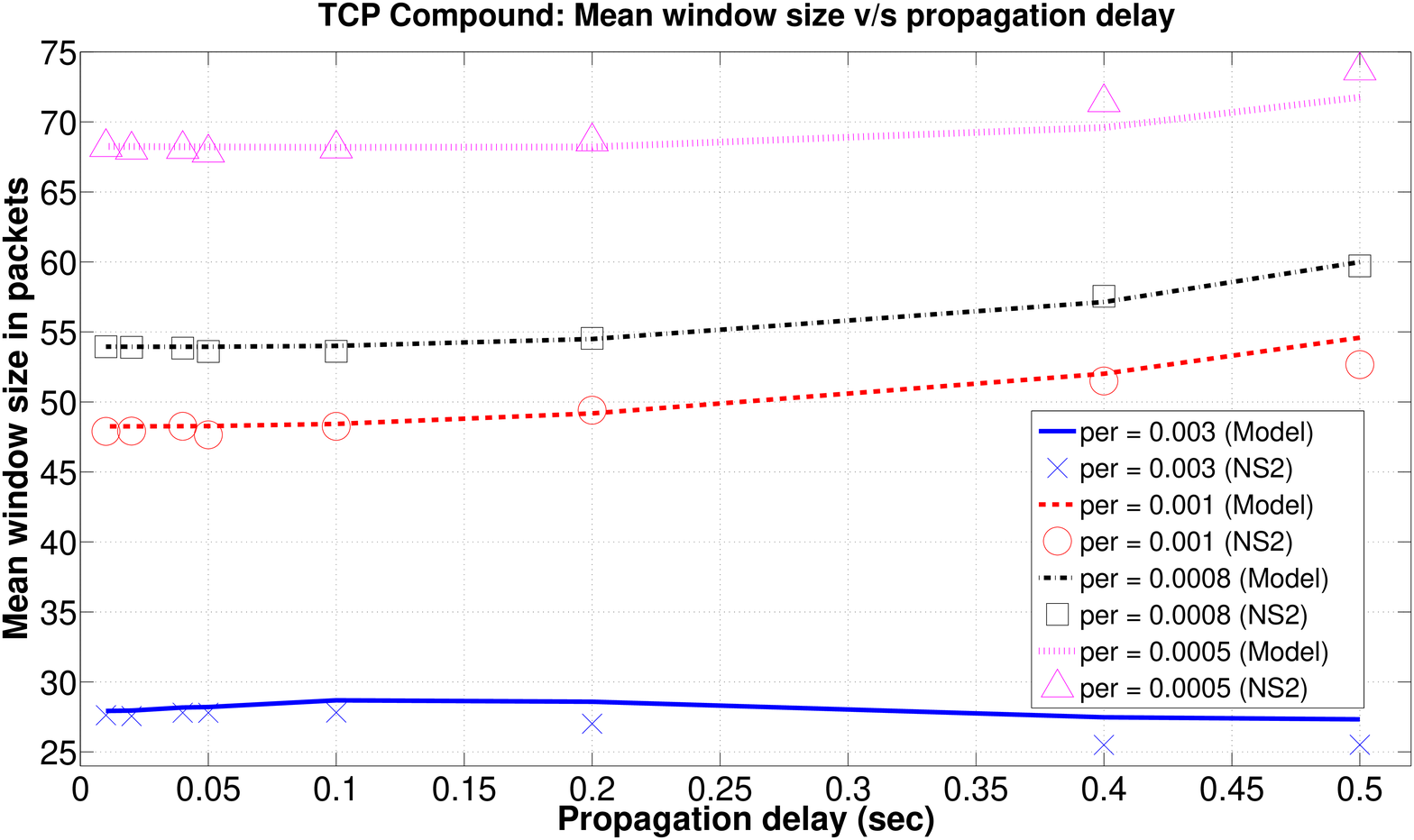}
  \caption{TCP Compound: Effect of change in BDP on $\mathbb{E}[W]$.}
  \label{fig:Effect_RTT_CTCP2}
\end{figure}

\begin{figure}
  \centering
  \includegraphics[scale=0.25, trim = 120 5 120 5, clip=true]{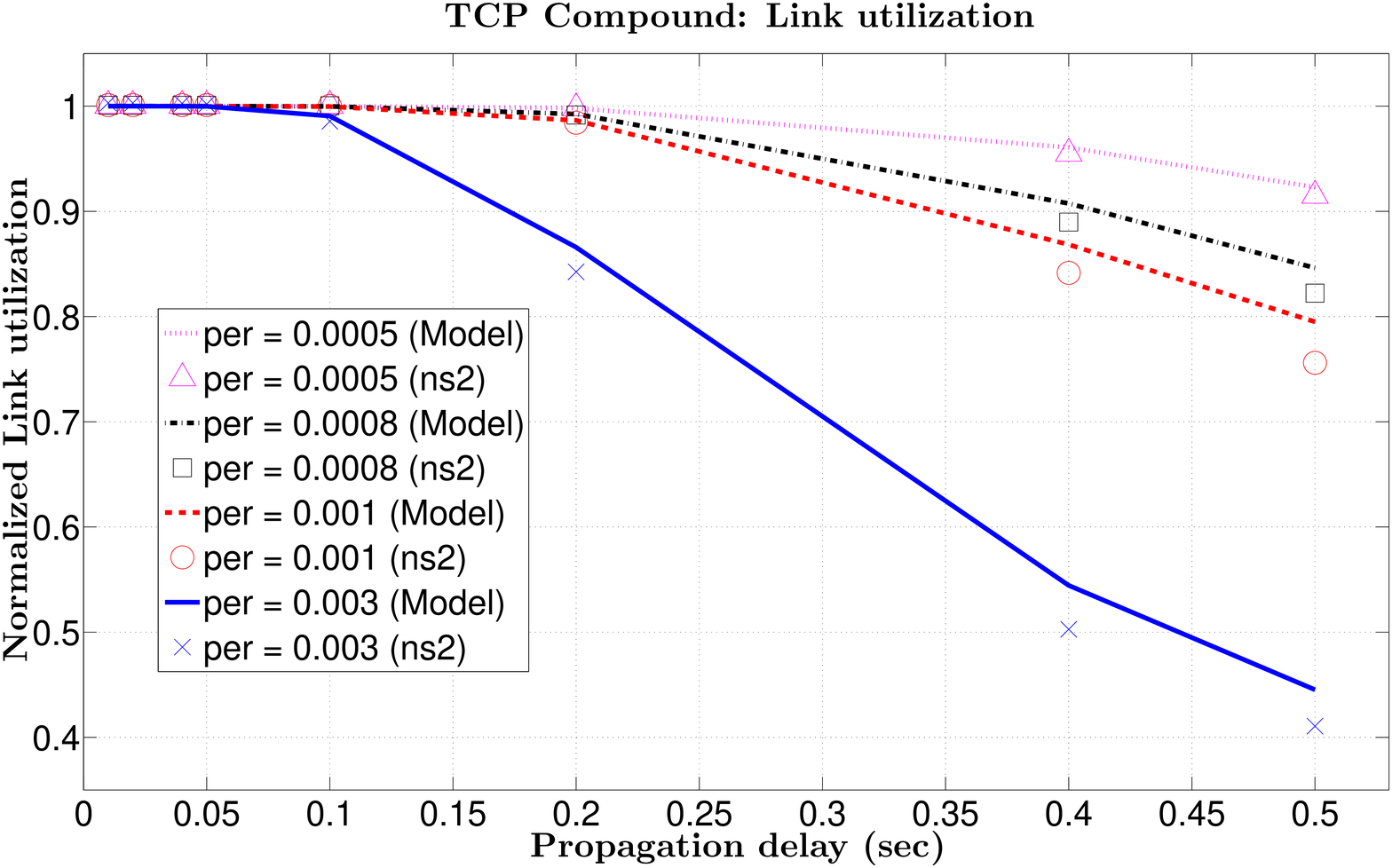}
  \caption{TCP Compound: link utilization.}
  \label{fig:CTCP_link_utilization}
\end{figure}

\begin{figure}
  \centering
  \includegraphics[scale=0.25, trim = 120 5 120 5, clip=true]{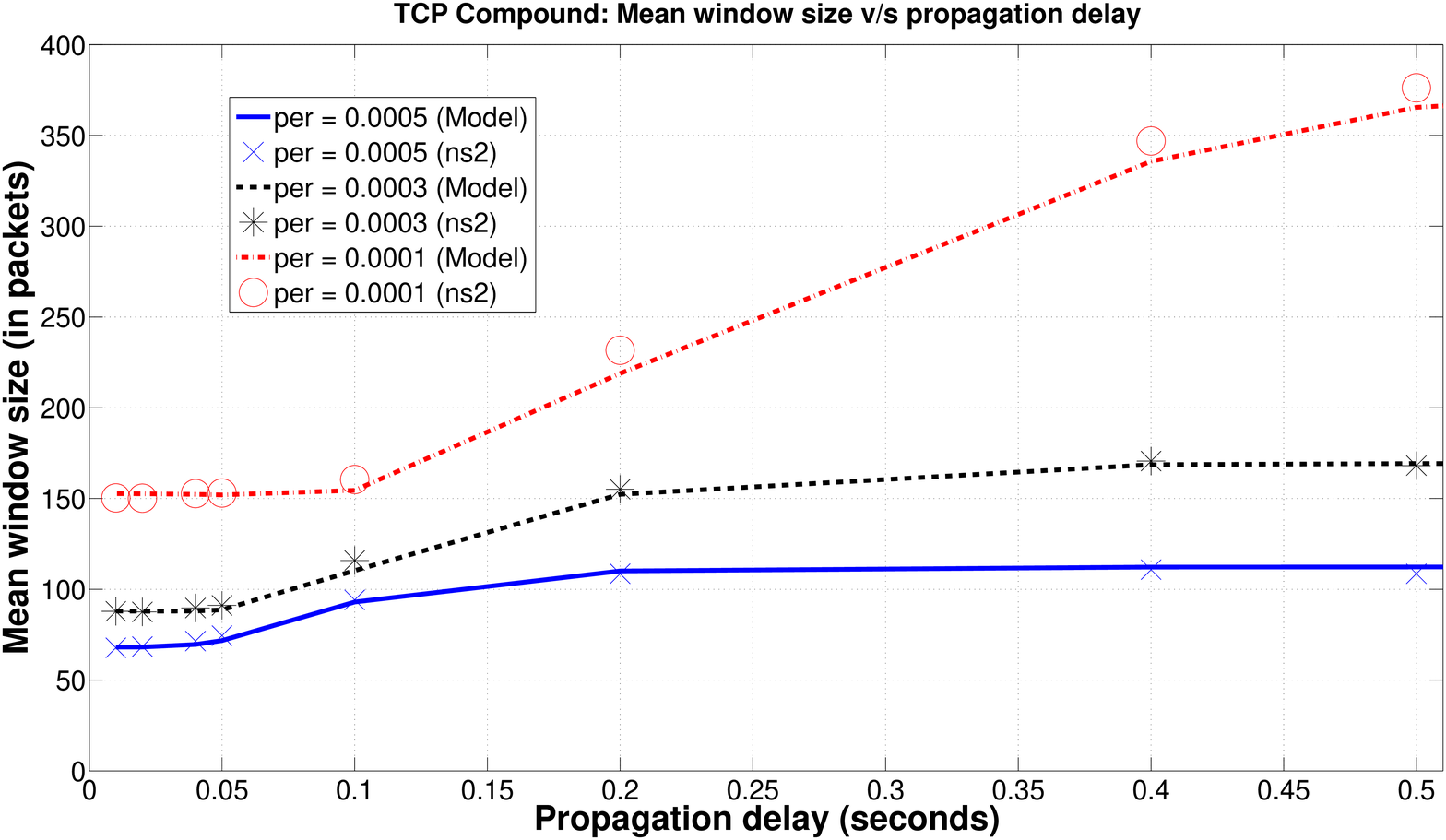}
  \caption{TCP Compound: Effect of change in BDP on $\mathbb{E}[W]$.}
  \label{fig:Effect_RTT_CTCP_3}
\end{figure}

\begin{figure}
  \centering
  \includegraphics[scale=0.25, trim = 120 5 120 5, clip=true]{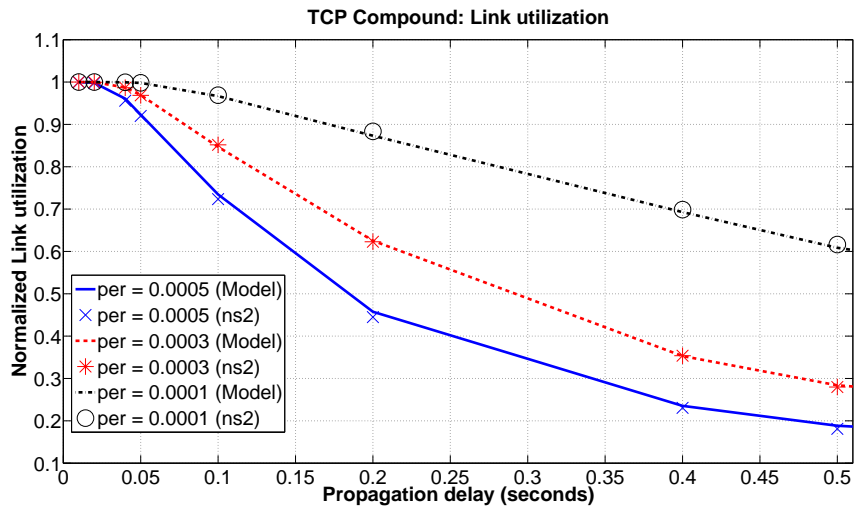}
  \caption{TCP Compound: link utilization.}
  \label{fig:CTCP_link_utilization_03}
\end{figure}


\section{Computational Complexity}
We note that Markovian models for TCP CUBIC and TCP Compound are also developed in \cite{Bao2010} and \cite{Blanc2009}. These papers look at the window sizes at the drop epochs and show that this process forms a discrete time Markov chain. Our approach looks at the window sizes at every RTT. The advantage of our approach is that it has reduced complexity. In our model every single-step transition of the Markov chain can lead to one of at most two states since a loss event or no-loss event can give us two potential next states. However if we consider the window size at drop epochs, any state in the state space could potentially be the next state in a single step transition. 

Assuming that the TCP Compound Markov chain in \cite{Blanc2009} has $N$ states, computation of the steady state probabilities, $\pi$ by use of iteration on $\pi = \pi P$ (where $P$ is the transition probability matrix), requires $O(N^2)$ steps (per iteration) for the model in \cite{Blanc2009}, whereas with our model it reduces to $O(N)$ steps.

For the TCP CUBIC model, assume that $W_{max} \leq N$. Then computation of the steady state probabilities, requires $O(N^2)$ steps (per iteration) for the model in \cite{Bao2010}. The Markov model that we have developed for TCP CUBIC requires $O(MN)$ steps (per iteration)  where $M$ is the maximum value that $T_n$ in equation \eqref{eqn:tcpCUBIC_MC_no_loss} takes. When window sizes are large, the TCP CUBIC mode in \eqref{eqn:tcpCUBIC_MC_no_loss} is dominant. In this case $\max \{T_n\}$ is $O(N^{\frac{1}{3}})$. Thus the Markov model for TCP CUBIC requires $O(N^{\frac{4}{3}})$ steps (per iteration). Hence our Markov models have lesser computational requirements. This can significantly help in computing the stationary distributions and $\mathbb{E}[W]$ for these flows because $N$ can be quite large. In \cite{Ha2008} and \cite{Tan2006Infocom}, the average window sizes go up to $10^6$ which would require $W_{max}$ to be of that order. In such cases, the computational savings will be huge.

\section{Multiple TCPs Through a General Network}
\label{sec:multiple_tcp_models}
In this section we use the models for a single TCP connection for TCP CUBIC and TCP Compound, that we described in Sections \ref{sec:cubic_markovmodel} and \ref{sec:ctcp_markovmodel} respectively, to study the behaviour of these TCP variants in the general network of Section \ref{sec:system_model} with multiple bottleneck links. We will also include TCP New-Reno connections for which the theoretical models are available in \cite{Padhye2000}, \cite{Mathis1997}. Now, the probability of error $p_r$ for flow $r$ will represent the end-to-end probability of packet errors which is seen by the TCP connection $r$. This may be the cumulative effect of packet overflows in the network and channel errors on wireless links on the route of flow $r$. We will use two techniques to include the effect of different interacting queues to obtain the throughput of different TCP flows.  The first technique approximates the behaviour of connecting links by M/G/1 queues and the second uses an optimization approach.

\subsection{The M/G/1 Approximation Method}
\label{sec:MG1_approx}
The main effect different TCP flows passing through the network have on each other is through the queuing delays they cause to each other which then affect their end-to-end RTT. We now describe an approximation to estimate queuing at the bottleneck links.

For a flow $r \in \mathcal{R}$, let $\mathbb{E}[W_r]$ be its average window size, $\lambda_r$ be its throughput (in packets/sec), $\mathbb{E}[R_r]$ its average RTT (including queuing delays) and $\mathbb{E}[s_r]$ the mean packet length in bits. We denote  the average queue length (in packets, of all connections passing through it, excluding the one being serviced) at link $l \in \mathcal{L}$ by $\mathbb{E}[Q_l]$ and the capacity of the link by $C_l$ (in bps). The average RTT for flow $r$ is given by Little's law, as 
\begin{equation} 
\label{eqn:RTT_General_1}
\begin{aligned}
\mathbb{E}[R_r] = \Delta_r & + \Bigl( \displaystyle{\sum_{l: A(l, r) = 1 } \frac{\mathbb{E}[Q_{l}]}{\sum_{r': A(l,r') = 1} \lambda_{r'}}} \Bigr)  \\
& + \Bigl( \displaystyle{\sum_{l: B(l, r) = 1 } \frac{\mathbb{E}[Q_{l}]}{\sum_{r': A(l,r') = 1} \lambda_{r'}}} \Bigr).
\end{aligned}
\end{equation}
The first term in RHS of \eqref{eqn:RTT_General_1} is the constant delay (propagation delay and transmission time), the second term denotes the sum of mean sojourn times for TCP packets on its route and the third term accounts for the total mean sojourn time for the ACK packets on their route. We will assume that the ACKs do not contribute  to the queuing, i.e., queuing at the bottleneck links is only caused by TCP packets. This is a reasonable assumption since typically ACK packets are much smaller than TCP packets.

We approximate the behaviour of the links by M/G/1 queues. The mean sojourn time at a link can then be found using the Pollaczek-Khinchine formula \cite{Bertsekas92}. The average queue length, at link $l$ is 
\begin{equation}
\label{eqn:EQ_MG1}
\mathbb{E}[Q_l] =  \frac{ (\lambda_l)^2 \mathbb{E}[s_l^2] } {2 C_l^2 (1 - \rho_l)},
\end{equation}
where $\lambda_l$ is the overall arrival rate at link $l$, $\mathbb{E}[s_l^2]$ is the second moment of the packet lengths at link $l$ and $\rho_l$ is the utilization factor of link $l$:

\begin{equation}
\label{eqn:tar_link_l}
\lambda_l = \sum_{r: A(l,r) = 1} \lambda_r,
\end{equation}

\begin{equation}
\label{eqn:packet_size_moments}
\mathbb{E}[s_l^2] = \sum_{r: A(l,r) = 1} \frac{\lambda_r}{\sum_{r: A(l,r) = 1} \lambda_r} \mathbb{E}[s_r^2],
\end{equation}

\begin{equation}
\label{eqn:rho_link_l}
\rho_l = \sum_{r: A(l,r) = 1}  \frac{\lambda_r \mathbb{E}[s_r]}{C_l}.
\end{equation}

The mean window size $\mathbb{E}[W_r]$ for connection $r$, its throughput, $\lambda_r$ and its RTT, $\mathbb{E}[R_r]$ are related by Little's law as, 
\begin{equation}
\label{eqn:littles_law}
\lambda_r = \frac{(1-p_r) \mathbb{E}[W_r] }{\mathbb{E}[R_r]}.
\end{equation}

We solve \eqref{eqn:RTT_General_1}, \eqref{eqn:EQ_MG1} and \eqref{eqn:littles_law} simultaneously for the unknowns $\overline{\lambda} = \{\lambda_r: r \in  \mathcal{R} \}$, $\{\mathbb{E}[W_r]: r \in  \mathcal{R} \}$, $\{\mathbb{E}[R_r]: r \in  \mathcal{R} \} $ and $\{\mathbb{E}[Q_l]: l \in  \mathcal{L} \} $ using \eqref{eqn:EW_TCP_CUBIC_MC} for TCP CUBIC, \eqref{eqn:EW_TCP_CTCP_MC} for TCP Compound and 
\begin{equation}
\label{eqn:EW_reno}
E[W_{reno}] = \frac{1.31}{\sqrt{p}}
\end{equation}
for TCP New Reno from \cite{Mathis1997}. 

We illustrate our approach briefly in Algorithm \ref{alg:MG1_approx}. In Algorithm \ref{alg:MG1_approx}, we iteratively find solution to the system of non-linear equations, \eqref{eqn:RTT_General_1}--\eqref{eqn:littles_law} using Broyden's algorithm \cite{Broyden1965} which is an efficient well known quasi-Newton method. The matrix $J$ is the Jacobian matrix of the fixed point equation, denoted by $f$. The parameter $t$ is identified using binary search along direction $d$. The average RTT, $\{\mathbb{E}[R_r]: r \in \mathcal{R}\}$ of the flows are computed using \eqref{eqn:RTT_General_1}. The average window sizes are computed using results from Sections \ref{sec:cubic_markovmodel} and \ref{sec:ctcp_markovmodel}. The procedure gives us the TCP performance measures, $\{\mathbb{E}[W_r], \lambda_r, \mathbb{E}[R_r]: r \in \mathcal{R}\}$.

The M/G/1 approximation technique we describe here is an approximation as the arrival process for TCP packets at a queue is not Poisson. This approximation helps us to easily compute the queuing delay at each link without prior identification of the bottleneck links. However the TCP dynamics have also been captured in equation \eqref{eqn:littles_law} via mean window lengths. Through simulation results described in Section \ref{sec:simulation_results}, we will see that these simplifying approximations are reasonably accurate.

\begin{algorithm}
  \begin{algorithmic}
    \STATE INPUT: $\{p_r, \Delta_r, \mathbb{E}[s_r]: r \in \mathcal{R}\}, \{C_l: l \in \mathcal{L}\}, A, B, \epsilon$
    \STATE OUTPUT: $\{\lambda_r, \mathbb{E}[W_r], \mathbb{E}[R_r]: r \in \mathcal{R}\}$
    \STATE Initialize $\overline{\lambda} = \{\lambda_r: r \in \mathcal{R}\}$
    \STATE Compute $ \{ \mathbb{E}[R_r](\overline{\lambda})$, $\mathbb{E}[W_r](\overline{\lambda}) : r \in \mathcal{R}\}$
    \STATE Compute 
    \begin{equation*}
    f(\overline{\lambda}) = \biggl\{ \lambda_r  - \frac{(1 - p_r)\mathbb{E}[W_r](\overline{\lambda})}{\mathbb{E}[R_r](\overline{\lambda})} \biggr\}_{r \in \mathcal{R}}
    \end{equation*}
    \STATE Compute Jacobian matrix, $J$ for $f$
    \STATE $H = J^{-1}$
    \WHILE {$\norm{f(\overline{\lambda})}_2 > \epsilon$}
    \STATE Set $d = - H f(\overline{\lambda})$     
    \STATE Choose $t$ such that $\norm{f(\overline{\lambda} + d t)}$ is minimized
    \STATE Set $\overline{\lambda} = \overline{\lambda} + d t$
    \STATE Compute $ \{ \mathbb{E}[R_r](\overline{\lambda})$, $\mathbb{E}[W_r](\overline{\lambda}) : r \in \mathcal{R}\}$
    \STATE Recompute $H$ using \cite{Broyden1965}
    \ENDWHILE
    \RETURN $\{\mathbb{E}[W_r], \lambda_r, \mathbb{E}[R_r]: r \in \mathcal{R}\}$
      \end{algorithmic}
  \caption{Throughput via M/G/1 approximation.}
  \label{alg:MG1_approx}
\end{algorithm}

\subsection{The Optimization Approach}
\label{sec:optimization}
In this section, we describe an optimization approach to compute the average throughput in the network. We will use the expressions \eqref{eqn:EW_TCP_CUBIC_MC} for TCP CUBIC, \eqref{eqn:EW_TCP_CTCP_MC} for TCP Compound and \eqref{eqn:EW_reno} for TCP Reno and use them in an optimization program. This approach is similar to the one used in \cite{Altman2002}. However, \cite{Altman2002} assumes that there is negligible queuing in the network and has only TCP Reno whereas in our model the queuing may be non-negligible and also has TCP Compound and TCP CUBIC connections.  

We will first consider a simple approximation when there is only one bottleneck queue. Suppose there are multiple TCP connections going through a single bottleneck link router. The TCP throughput for the different connections can be computed using the following approximation. Let $C$ be the bottleneck link capacity (in bps) and let $D$ be the average queuing delay at the bottleneck link. We can compute the average window size for TCP CUBIC, TCP Compound and New Reno using equations \eqref{eqn:EW_TCP_CUBIC_MC}, \eqref{eqn:EW_TCP_CTCP_MC} and  \eqref{eqn:EW_reno} respectively. If we have $\sum_{r} \frac{(1 - p_r)\mathbb{E}[W_r]\mathbb{E}[s_r]}{\Delta_r} \leq C$, we set the throughput (in packets/sec) for connection $r$ as $\lambda_r = \frac{(1 - p_r) \mathbb{E}[W_r]}{\Delta_r}$ and $M=0$, else we find $M$ such that
\begin{equation}
\sum_{r} \frac{(1 - p_r)\mathbb{E}[W_r]\mathbb{E}[s_r]}{\Delta_i + M} = C.
\end{equation}

In the above case we have assumed that the bottleneck link is either operating at full capacity or has zero queuing. Extending this for the multihop network, we assume that if $\sum_{r: A(l,r) = 1} \lambda_r \mathbb{E}[s_r] < C_l$ then $\mathbb{E}[Q_l] = 0$. Therefore, for each $r \in \mathcal{R}$ and each $l \in \mathcal{L}$, we have
\begin{equation}
\label{eqn:LL}
(1 - p_r)\mathbb{E}[W_r] = \lambda_{r} (\Delta_r + \sum_{l: A(l,r) = 1} M_l), 
\end{equation}
and 
\begin{equation}
\label{eqn:slackness}
M_l (C_l - \sum_{r: A(l,r) = 1}\lambda_r E[s_r] ) = 0, 
\end{equation}
where $M_l$ is the mean sojourn time at link $l$. The equation \eqref{eqn:LL} comes from the Little's law, whereas \eqref{eqn:slackness} is a restatement of our assumption. Besides \eqref{eqn:LL} and \eqref{eqn:slackness}, we require 
\begin{equation}
\label{eqn:capacity_constraints}
\sum_{r: A(l,r) = 1}\lambda_r \mathbb{E}[s_r] \leq C.
\end{equation}

Assuming $\mathbb{E}[W_r]$ is known, we can solve \eqref{eqn:LL}--\eqref{eqn:capacity_constraints} (for the unknowns $\{\lambda_r \in \mathcal{R}\}$ and $\{M_l \in \mathcal{L}\}$) using Algorithm \ref{alg:KKT_solver}. The algorithm stops if the successive approximations to $M^*$ are $\epsilon-$close. The parameter $d$ in the algorithm is the step size which dictates the speed of convergence.
\begin{algorithm}
  \begin{algorithmic}
    \STATE INPUT: $\{\mathbb{E}[W_r], p_r, \Delta_r, \mathbb{E}[s_r]: r \in \mathcal{R}\}, \{C_l: l \in \mathcal{L}\}, A, \epsilon$.
    \STATE OUTPUT: $M^{*} = \{M_l^*: l \in \mathcal{L}\}, \lambda^* = \{\lambda_r^*: r \in \mathcal{L}\}$.
    \STATE Initialize $M^{1} = \{M_l^1: l \in \mathcal{L}\}$.
    \STATE Initialize $M^{2} = \{M_l^2: l \in \mathcal{L}\}$.
    \STATE Set step size $d$.
    \REPEAT
    \STATE $M^{1} = M^{2}$.
    \FORALL{ $r \in \mathcal{R}$}
    \STATE $\lambda_r = \frac{(1 - p_r)\mathbb{E}[W_r]}{\Delta_r + \displaystyle{\sum_{j: A(j,r) = 1} M_j}}$. 
    \ENDFOR
    \FORALL{ $l \in \mathcal{L}$}
    \STATE $M_l^2 = (M_l^1 - d(C_l - \displaystyle{\sum_{r: A(l,r) = 1}\lambda_r \mathbb{E}[s_r]))^+}$
    \ENDFOR
    \UNTIL{$\norm{M^{1} - M^{2}} \leq \epsilon$}
    \STATE $M^* = M^{2}$
    \STATE $\lambda^* = \{\lambda_r: r \in \mathcal{R}\}$
    \RETURN $M^*, \lambda^*.$
      \end{algorithmic}
  \caption{Solution to \eqref{eqn:LL}--\eqref{eqn:capacity_constraints}.}
  \label{alg:KKT_solver}
\end{algorithm}

The average window size expressions (as computed using the Markov chain models in Sections \ref{sec:cubic_markovmodel} and \ref{sec:ctcp_markovmodel}) are a function of the average RTT (for TCP CUBIC) and average queue size (for TCP Compound) of the connections. Thus, we need to solve for the average window size equations, and equations \eqref{eqn:LL} and \eqref{eqn:slackness} simultaneously. As in Section \ref{sec:MG1_approx}, we use Broyden's algorithm for simultaneously solving these equations and obtaining the throughput of the different connections. We describe our approach in brief in Algorithm \ref{alg:optimization_approach}. In the M/G/1 approach, we initialize Algorithm \ref{alg:MG1_approx} with a guess of the throughput of the different flows. Here, we initialize Algorithm \ref{alg:optimization_approach}, with a guess of $\overline{X}$ which consists of $\{\mathbb{E}[R_r]: r \in \mathcal{R} \text{ and } r \text{ is a CUBIC connections} \}$ and $\{\mathbb{E}[Q_r]: r \in \mathcal{R} \text{ and } r \text{ is a Compound connections} \}$ so that we can compute $\{ \mathbb{E}[W_r] : r \in \mathcal{R} \}$. We then compute the throughputs, $\{\lambda_r: r \in \mathcal{R}\}$, of the different connections and the average queue sizes, $\{M_l: l \in \mathcal{L}\}$, using Algorithm \ref{alg:KKT_solver}. Since in this approach, we did not account for the queuing incurred by the ACK packets, we need to update, $\Delta_r$ and hence $\{ \mathbb{E}[R_r] : r \in \mathcal{R} \}$ iteratively using the queuing estimates $\{ M_l: l \in \mathcal{L} \}$. We then iteratively find solution to the equation $g(\overline{X}) = 0$ (see \eqref{eqn:gofX_optimization}) using Broyden's algorithm \cite{Broyden1965}. The procedure gives us the TCP performance measures, $\{\mathbb{E}[W_r], \lambda_r, \mathbb{E}[R_r]: r \in \mathcal{R}\}$.

\begin{algorithm}
  \begin{algorithmic}
    \STATE INPUT: $\{p_r, \Delta_r, \mathbb{E}[s_r]: r \in \mathcal{R}\}, \{C_l: l \in \mathcal{L}\}, A, B, \epsilon$
    \STATE OUTPUT: $\{\lambda_r, \mathbb{E}[W_r], \mathbb{E}[R_r]: r \in \mathcal{R}\}$
    \STATE Initialize $\overline{X} = \{\mathbb{E}[R_s], \mathbb{E}[Q_t]: s,t \in \mathcal{R} \}$, where connections $s$, $t$  index TCP CUBIC and TCP Compound connections respectively.
    \STATE Compute $\{\mathbb{E}[W_r](\overline{X}) : r \in \mathcal{R}\}$
    \STATE Compute $\{ \lambda_r : r \in \mathcal{R}\}$, $\{ M_l : l \in \mathcal{L}\}$  using Algorithm \ref{alg:KKT_solver}
    \STATE Update $\{\Delta_r, \mathbb{E}[R_r] : r \in \mathcal{R}\}$ using $\{ M_l : l \in \mathcal{L}\}$
    \STATE Compute 
    \begin{equation}
    \label{eqn:gofX_optimization}
    g(\overline{X}) = \biggl\{ \lambda_r(\overline{X})  - \frac{(1 - p_r)\mathbb{E}[W_r](\overline{X})}{\mathbb{E}[R_r]} \biggr\}_{r \in \mathcal{R}}
    \end{equation}
    \STATE Compute Jacobian matrix, $J$ for $g$
    \STATE $H = J^{-1}$
    \WHILE {$\norm{g(\overline{X})}_2 > \epsilon$}
    \STATE Set $d = - H f(\overline{\lambda})$     
    \STATE Choose $t$ such that $\norm{g(\overline{X} + d t)}$ is minimized
    \STATE Set $\overline{X} = \overline{X} + d t$
    \STATE Compute $\{\mathbb{E}[W_r](\overline{X}) : r \in \mathcal{R}\}$
    \STATE Compute $\{ \lambda_r : r \in \mathcal{R}\}$ using Algorithm \ref{alg:KKT_solver}
    \STATE Update $\{\Delta_r, \mathbb{E}[R_r] : r \in \mathcal{R}\}$ using $\{ M_l : l \in \mathcal{L}\}$
    \STATE Recompute $H$ using \cite{Broyden1965}
    \ENDWHILE
    \RETURN $\{\mathbb{E}[W_r], \lambda_r, \mathbb{E}[R_r]: r \in \mathcal{R}\}$
      \end{algorithmic}
  \caption{Throughput via Optimization Approach.}
  \label{alg:optimization_approach}
\end{algorithm}

The optimization problem, provided in Proposition \ref{prop:optimization_trick_statement} is another way to solve \eqref{eqn:LL}--\eqref{eqn:capacity_constraints}. In \cite{Massoulie2002}, a similar optimization approach was used to obtain the rates attained by congestion control algorithms with fixed end-to-end window sizes in a network of queues which use FIFO scheduling. We use the following optimization program for different TCP variants with varying window sizes with finite means.

\begin{proposition} \label{prop:optimization_trick_statement} In the general system model, under the above assumption, (i.e., if $\sum_{r: A(l,r) = 1} \lambda_r < C_l$ then  $\mathbb{E}[Q_l] = 0$) the primal solutions of the optimization problem 
\begin{equation}
\label{eqn:optimization_problem}
\max \sum_{r \in \mathcal{R}} \mathbb{E}[s_r] \Bigl((1 - p_r) \mathbb{E}[W_r] \log(\lambda_r) - \lambda_r \Delta_r \Bigr)
\end{equation}
such that
\begin{equation*}
\lambda_r \geq 0, \forall r \in \mathcal{R} \text{ and } \sum_{r: A(l,r) = 1}\lambda_r \mathbb{E}[s_r] \leq C_l,  \forall l \in \mathcal{L}, 
\end{equation*}
provide the system throughputs, $\{\lambda_r: r \in \mathcal{R}\}$.
\end{proposition}

\begin{proof}
Consider the optimization problem \eqref{eqn:optimization_problem}. Let $\{\lambda_r^{*}: r \in \mathcal{R}\}$ denote the optimal point for \eqref{eqn:optimization_problem} and let us denote the dual optimal points corresponding to the capacity constraints by $\{M_l^{*}: l \in \mathcal{L}\}$ and the dual optimal points corresponding to the non-negativity constraints by $\{\gamma_r^{*}: r \in \mathcal{R}\}$. The KKT conditions \cite{Boyd2004} for the above problem are given as follows. For each $r \in \mathcal{R}$ and each $l \in \mathcal{L}$, we have
\begin{subequations}
\label{eqn:KKT}
\begin{flalign}
& (1 - p_r)\mathbb{E}[W_r] = \lambda_{r}^{*} (\Delta_r + \sum_{l: A(l,r) = 1} M_l^{*} + \frac{\gamma_{r}^{*}}{\mathbb{E}[s_r]}),  \label{eqn:KKT_1}\\
& \hspace{15mm} M_l^{*} (C_l - \sum_{r: A(l,r) = 1}\lambda_r^{*} \mathbb{E}[s_r] ) = 0, \label{eqn:KKT_slack_1}\\
& \hspace{35mm} \gamma_r^{*}\lambda_r^{*} = 0 \label{eqn:KKT_slack_2}.
\end{flalign}
\end{subequations}
Also $M_l^{*} \geq 0$, for each $l$ and $\lambda_r^{*} \geq 0$, for each $r$. The objective function is concave and $\lambda_r = 0$, $\forall r \in \mathcal{R}$ is strictly feasible. Therefore by Slater's condition \cite{Boyd2004}, strong duality holds. Hence any pair of primal optimal and dual optimal points must satisfy the KKT conditions listed in \eqref{eqn:KKT}. Since the problem is concave, the KKT conditions are also sufficient for optimality. Therefore we conclude that any solution to \eqref{eqn:optimization_problem} also satisfies \eqref{eqn:KKT} and vice versa. 

For each $r \in \mathcal{R}$, the term $\mathbb{E}[W_r]$ in \eqref{eqn:LL} is greater than $0$ (since TCP window sizes are $\geq 1$). Thus, for any solution to \eqref{eqn:LL}--\eqref{eqn:capacity_constraints},  $\lambda_r > 0$ for each $r\in \mathcal{R}$. Therefore under the above assumptions, any $\{\lambda_r > 0: r \in \mathcal{R}\}$ and $\{M_l \geq 0: l \in \mathcal{L}\}$ which satisfies \eqref{eqn:LL} and \eqref{eqn:slackness} satisfies the KKT conditions \eqref{eqn:KKT} with $\gamma_r^{*} = 0$ for all $r \in \mathcal{R}$. Therefore $\lambda_r$ for all $r \in \mathcal{R}$ and $M_l$ for all $l \in \mathcal{L}$ which satisfy \eqref{eqn:LL}, \eqref{eqn:slackness} and \eqref{eqn:capacity_constraints} also satisfy the KKT conditions.
\end{proof}

The optimization program \eqref{eqn:optimization_problem} is a strict concave optimization problem with affine constraints and hence has a unique solution. Thus this optimization problem can be solved via the usual convex-optimization algorithms \cite{Boyd2004}. We note that as strong duality holds for \eqref{eqn:optimization_problem}, we can solve it via its dual. Algorithm \ref{alg:KKT_solver} solves the optimization program \eqref{eqn:optimization_problem} via $\{M_l: l \in \mathcal{L}\}$, the dual variables. It is a projected gradient method used to minimize the dual of optimization program \eqref{eqn:optimization_problem}. In Algorithm \ref{alg:KKT_solver}, the dual variable, $M_l$, i.e., the queuing delay at link $l$ is reduced whenever the arrival rate $\displaystyle{\sum_{r: A(l,r) = 1}\lambda_r \mathbb{E}[s_r]}$ at link $l$ exceeds capacity $C_l$ and is increased whenever the arrival rate $\displaystyle{\sum_{r: A(l,r) = 1}\lambda_r \mathbb{E}[s_r]}$ at link $l$ is less than capacity $C_l$. If we choose $d$ small enough, Algorithm \ref{alg:KKT_solver} converges after either setting $M_l$ close to $0$  or when $\displaystyle{\sum_{r: A(l,r) = 1}\lambda_r \mathbb{E}[s_r]}$ at link $l$ almost equals the capacity $C_l$, for all $l \in \mathcal{L}$. Similar techniques for solving optimization programs have been used in \cite{Krishna2013} and \cite{Palomar2006}. Through simulation results described in Section \ref{sec:simulation_results}, we will see that our approximations made in this approach are reasonable.
\section{Simulation Results}
\label{sec:simulation_results}

In this section, we will validate our models using ns2 simulations. We use extensions from \cite{Wei2006NS2} for high speed TCP which use actual Linux TCP code for the simulations. We first look at the results for single bottleneck link.

\subsection{Single Bottleneck Link}
We first consider the case when multiple flows go through a single bottleneck link and the remaining connecting links have sufficiently high capacity. We use the M/G/1 approximation based model and the optimization approach for a single bottleneck queue discussed earlier in Sections \ref{sec:MG1_approx} and \ref{sec:optimization} respectively. In this setup, we have $12$ flows ($4$ Compound, $4$ CUBIC, $4$ New Reno) going through a single bottleneck link. We number the flows $1-12$, the flows indexed $i$ such that $\mod(i,3) = 1$\footnote{ $mod(i,j) = i \% j$ } are TCP Compound flows, $\mod(i,3) = 2$ are TCP CUBIC and the rest are TCP New Reno. These flows have different propagation delays. The first three flows have propagation delay of $0.01$ sec, flows $4-6$ have delay of $0.02$ sec, flows $7-9$ have delay of $0.1$ sec and the last three flows have propagation delay of $0.2$ sec. All flows have TCP packet size of $1050$ bytes which is the default packet size of ns2. The packet error rates for all flows is $0.001$. We consider three cases with different bottleneck link capacities. In Figure \ref{fig:per_001_50M_1N_12F}, the bottleneck capacity is $50$ Mbps which illustrates the situation when there is severe congestion. Figures \ref{fig:per_001_100M_1N_12F}, \ref{fig:per_001_1G_1N_12F} have bottleneck capacities of $100$ Mbps and $1$ Gbps illustrating moderate and low congestion scenario respectively. Our theoretical results are quite close to the simulation results and all errors are less than $10\%$. In Table \ref{tbl:avg_queue_singleBL}, we compare results for the average queuing at the bottleneck link and the bottleneck link utilization.

From the Figures \ref{fig:per_001_50M_1N_12F}, \ref{fig:per_001_100M_1N_12F} and \ref{fig:per_001_1G_1N_12F}, we see that among the different TCP versions with same propagation delay, TCP Compound (flows with $\mod(i,3) = 1$) gets the highest throughput when the propagation delay is small ($\Delta \leq 0.1$ sec) whereas TCP CUBIC (flows with $\mod(i,3) = 1$) gets the highest throughput when propagation delay is large ($\Delta = 0.2$ sec). However as link speeds reduce from $1$ Gbps to $50$ Mbps we see that the different TCP versions have a more equitable share of the bottleneck link capacity, for $\Delta \leq 0.1$ sec. This happens due to the increase in RTT, due to queuing, when the link speed is reduced from $1$ Gbps to $50$ Mbps which causes (a) increase in average window size of TCP CUBIC as for fixed PER its window size grows with RTT and (b) decrease in average window size of TCP Compound which decreases with increase in queue size.

\begin{figure}
  \centering
  \includegraphics[scale=0.25, trim = 80 5 140 5, clip=true]{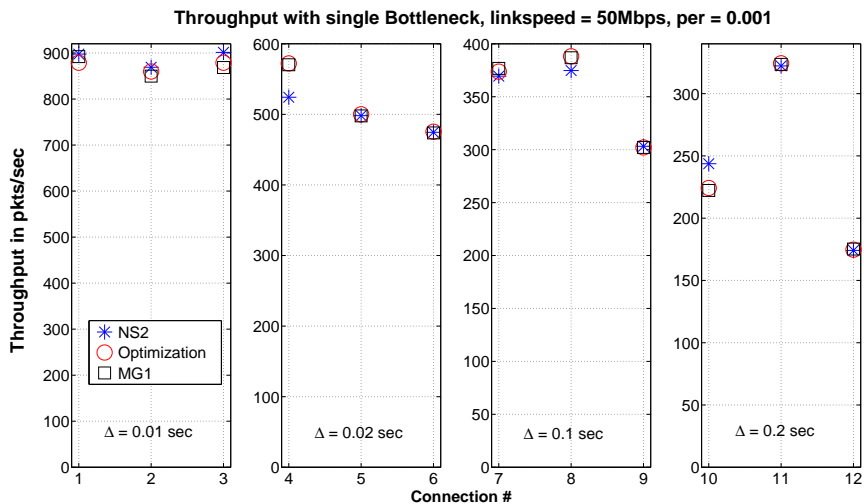}
  \caption{Single bottleneck with link capacity $50$ Mbps.}
  \label{fig:per_001_50M_1N_12F}
\end{figure}

\begin{figure}
  \centering
  \includegraphics[scale=0.25, trim = 80 5 140 5, clip=true]{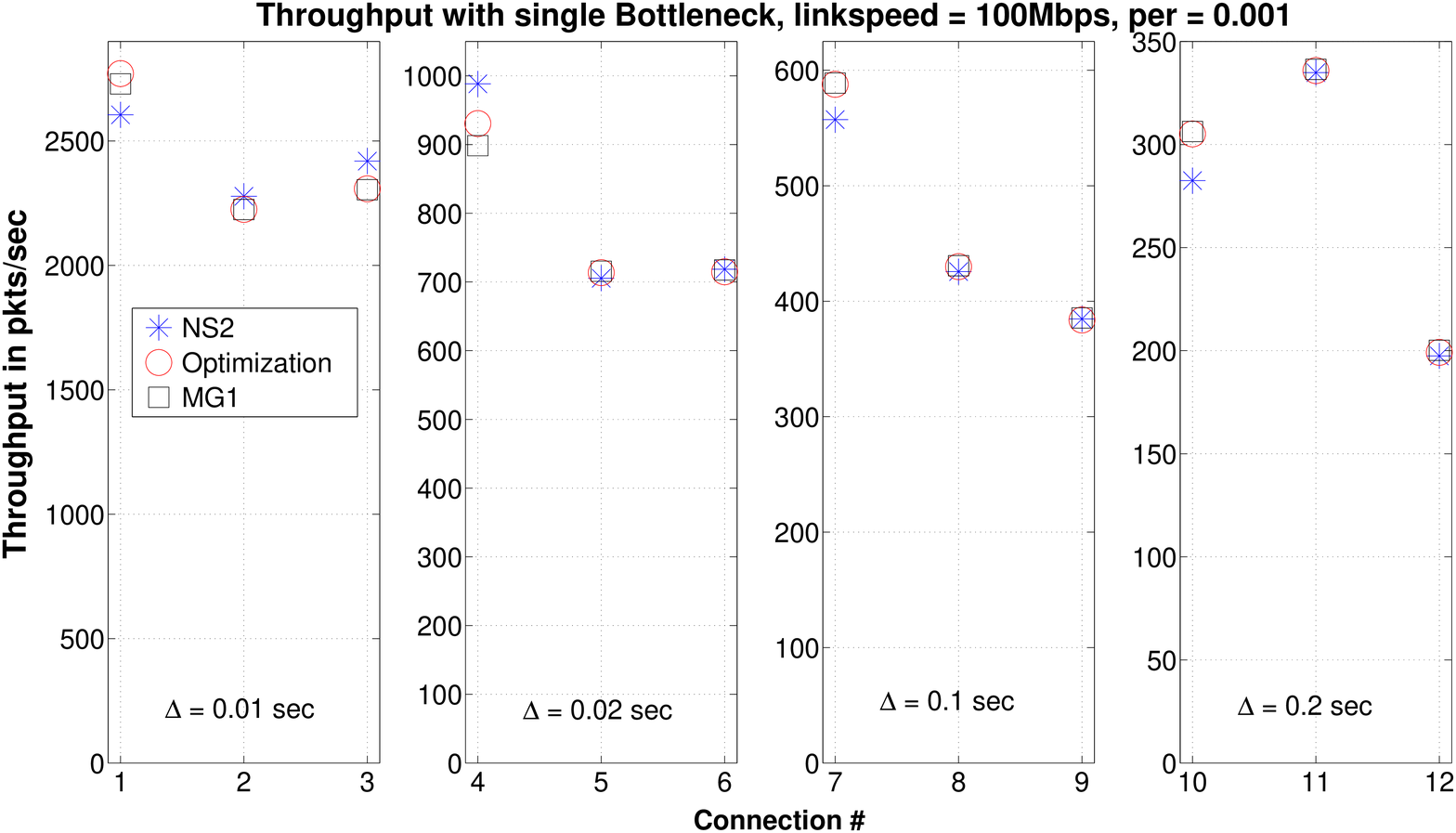}
  \caption{Single bottleneck with link capacity $100$ Mbps.}
  \label{fig:per_001_100M_1N_12F}
\end{figure}

\begin{figure}
  \centering
  \includegraphics[scale=0.25, trim = 80 5 140 5, clip=true]{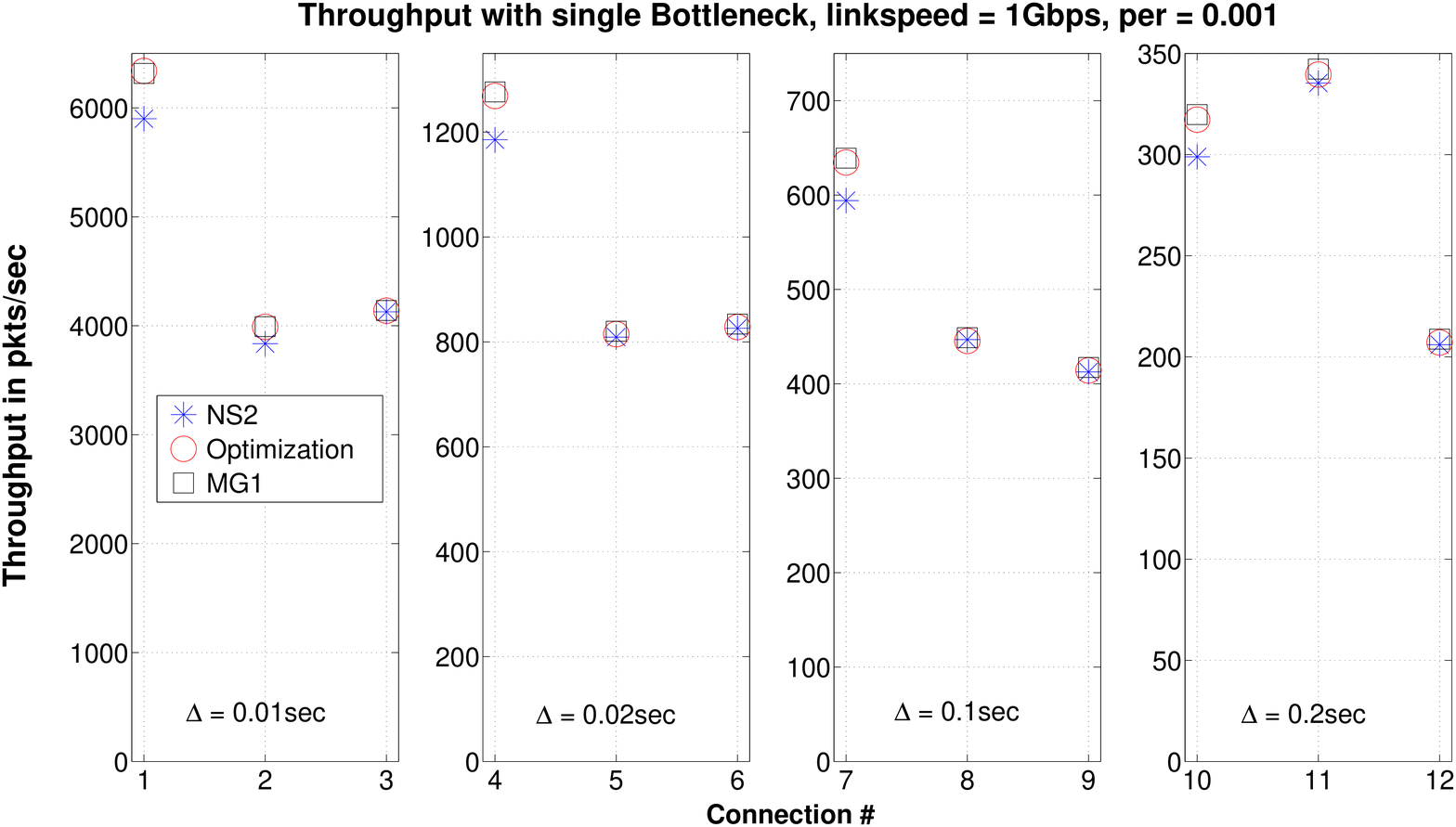}
  \caption{Single bottleneck with link capacity $1$ Gbps.}
  \label{fig:per_001_1G_1N_12F}
\end{figure}

\begin{table}
\centering
\caption{Single bottleneck: average queue size, link utilization at bottleneck queue.}
\scalebox{0.9}{
\begin{tabular}{|c|c|c|c|c|c|c|}
\hline
Link Speed & \multicolumn{3}{|c|}{Average Queue Size (in packets)} & \multicolumn{3}{|c|}{Normalized Link utilization}\\ \cline{2-7}
& ns2 & Approx. Model & MG1 & ns2 & Approx. Model & MG1 \\ 
\hline
$50$Mbps & 223.2& 220.2 & 217.4 & 1.0 & 1.0 & 0.997 \\
$100$Mbps & 92.4 & 94.0 & 87.3 & 0.999 & 1.0 & 0.994 \\
$1$Gbps & 0.17 & 0.0 & 0.016 & 0.159 & 0.166 & 0.166\\
\hline
\end{tabular}
}
\label{tbl:avg_queue_singleBL}
\end{table}

\subsection{Multiple Bottleneck Links}
In this section, we consider examples with multiple bottleneck links. We consider the case when besides TCP packets, ACKs also get delayed. In the first example in this scenario, we have ten routers and a total of $15$ flows, $5$ using TCP Compound, $5$ using CUBIC and $5$ using Reno. We group the flows into $5$ flow groups, ($F_1-F_5$) with each group consisting of $3$ flows with one flow each of TCP Compound, CUBIC and Reno. The network topology is shown in Figure \ref{fig:10N_15F_diag}. The links are bidirectional and are symmetric, i.e., forward direction and reverse direction have the same bandwidth and propagation delay. We note that in this example, the ACKs may also be subject to queuing which is accounted for in both the M/G/1 based models and the optimization based discussed in Sections \ref{sec:MG1_approx} and \ref{sec:optimization}. 

In Figures \ref{fig:10N_15F_case1A} and \ref{fig:10N_15F_case9A} we compare the throughputs obtained by the different flows for two different configurations with network topology given in Figure \ref{fig:10N_15F_diag}. In Figures \ref{fig:10N_15F_case1A_util} and \ref{fig:10N_15F_case9A_util} we compare the link utilization for links with link utilization $> 70 \%$ (in ns2 simulations) in these two configurations. The two configurations have different link speeds (mentioned in Figures \ref{fig:10N_15F_case1A_util} and \ref{fig:10N_15F_case9A_util}) and we denote them as config. 1 and config. 2. All flows in a flow group have the same propagation delay and packet error rate. Flows are numbered $1-15$. Flows numbered $i$ such that $\mod(i,3) = 1$ are TCP Compound flows, flows with $\mod(i,3) = 2$ are TCP CUBIC whereas rest are TCP New Reno flows. The parameters used for simulation are mentioned in Figures \ref{fig:10N_15F_case1A} and \ref{fig:10N_15F_case9A}. The link speeds are mentioned in Mbps and the propagation delays are mentioned in sec. All flows have packet sizes $1050$ bytes which is the default value in ns2. For the first configuration, for the throughput obtained by the different flows, the M/G/1 model results differ from the ns2 simulations by less than $8\%$ whereas for the optimization approach the maximum difference is $10.5\%$. In the second configuration, for the throughput obtained by the different flows, both models differ from the ns2 simulations by less than $8\%$.
From \ref{fig:10N_15F_case1A_util} and \ref{fig:10N_15F_case9A_util}, we see that both the M/G/1 and the optimization techniques identify the bottleneck links (links with utilization  $ > 70\%$) correctly. The difference in link utilization, as compared to ns2 simulations, for all the links in the network, is less than $4\%$ for the M/G/1 technique in both the configurations. For the optimization technique, the difference from ns2 simulations is less than $7\%$.

\begin{figure}
  \centering
  \includegraphics[scale=0.8]{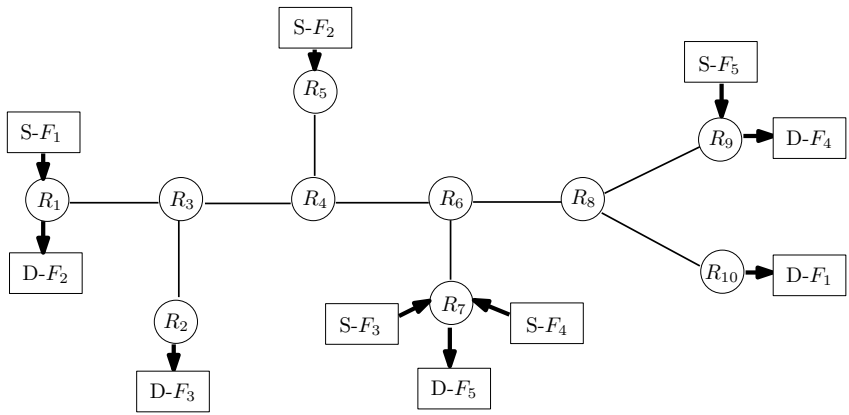}
  \caption{Topology of scenario with 10 Routers, 15 Flows.}
  \label{fig:10N_15F_diag}
\end{figure}

\begin{figure}
  \centering
  \includegraphics[scale=0.25, trim = 80 5 140 5, clip=true]{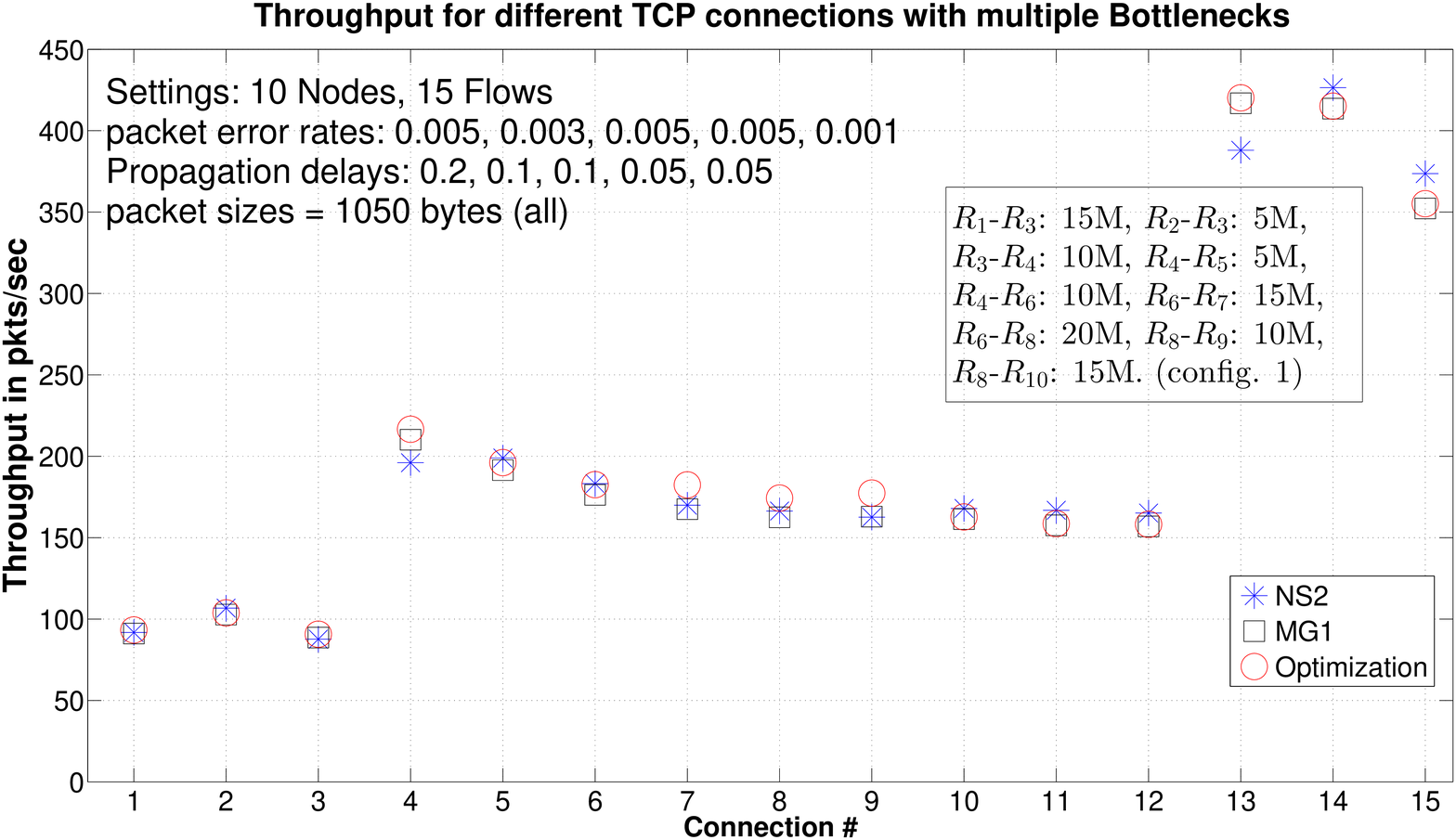}
  \caption{Throughputs for example in Figure \ref{fig:10N_15F_diag}, config. 1.}
  \label{fig:10N_15F_case1A}
\end{figure}

\begin{figure}
  \centering
  \includegraphics[scale=0.25, trim = 80 5 140 5, clip=true]{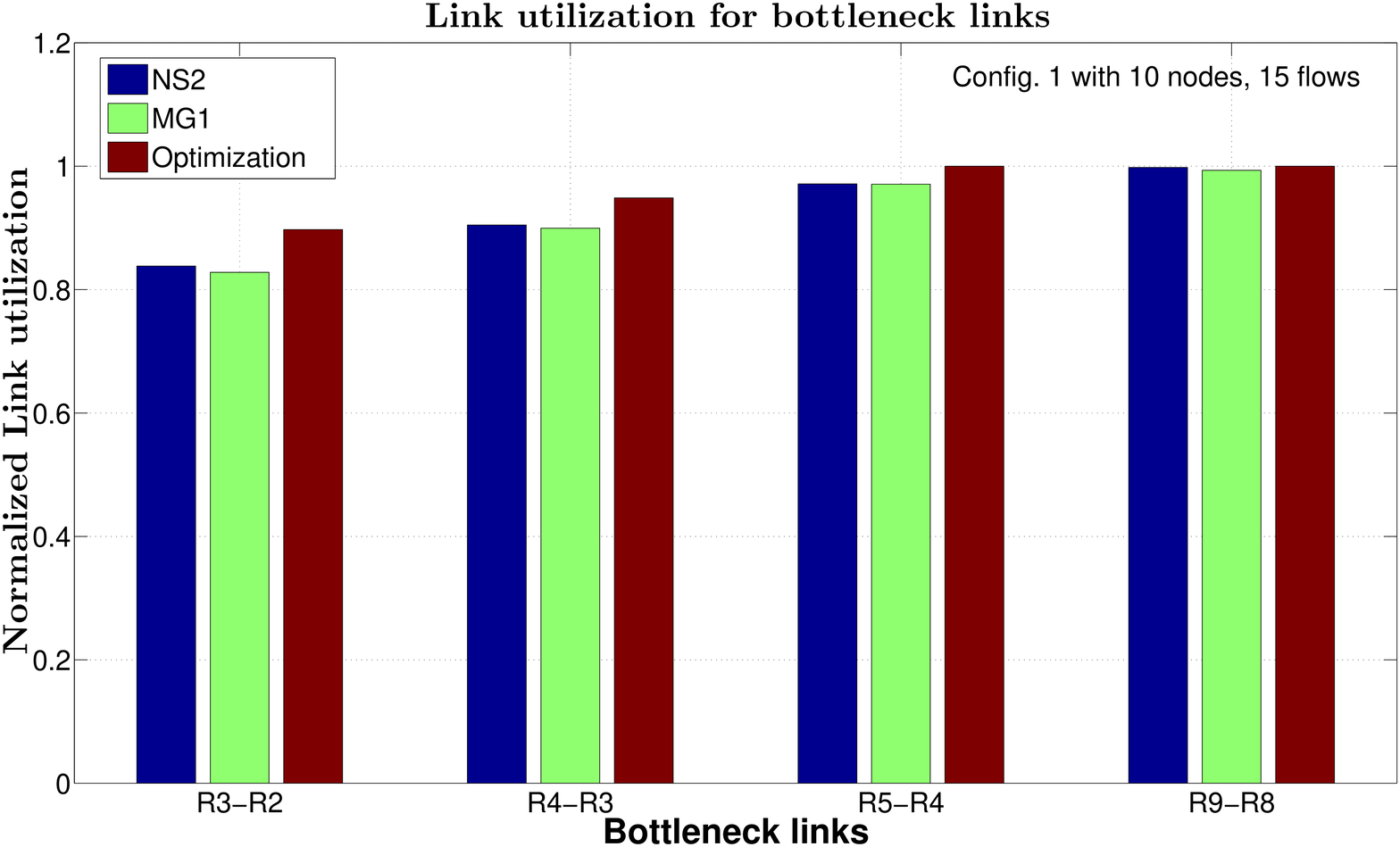}
  \caption{Link utilization for example in Figure \ref{fig:10N_15F_diag}, config. 1.}
  \label{fig:10N_15F_case1A_util}
\end{figure}

\begin{figure}
  \centering
  \includegraphics[scale=0.25, trim = 80 5 140 5, clip=true]{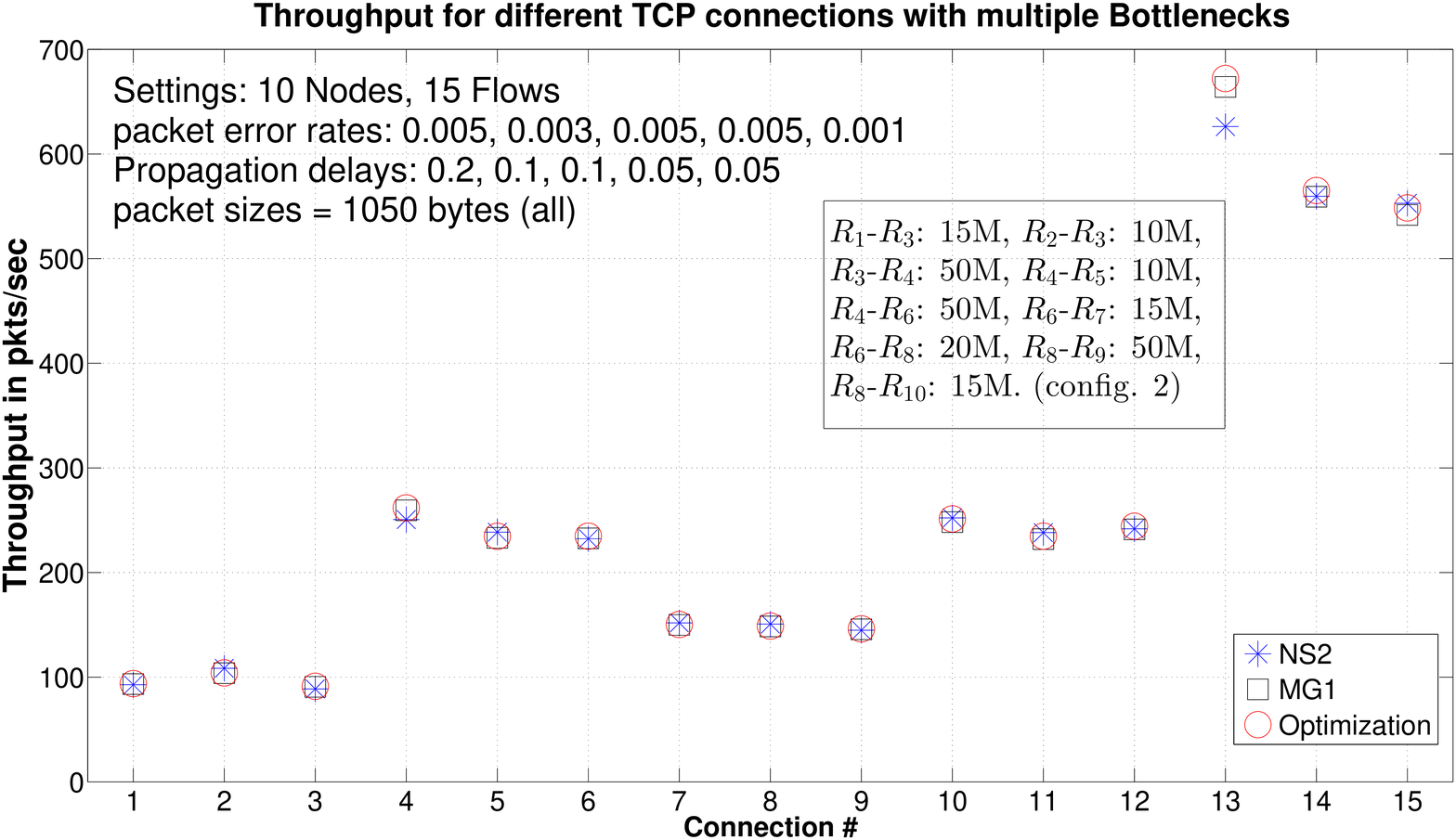}
  \caption{Throughputs for example in Figure \ref{fig:10N_15F_diag}, config. 2.}
  \label{fig:10N_15F_case9A}
\end{figure}

\begin{figure}
  \centering
  \includegraphics[scale=0.25, trim = 80 5 140 5, clip=true]{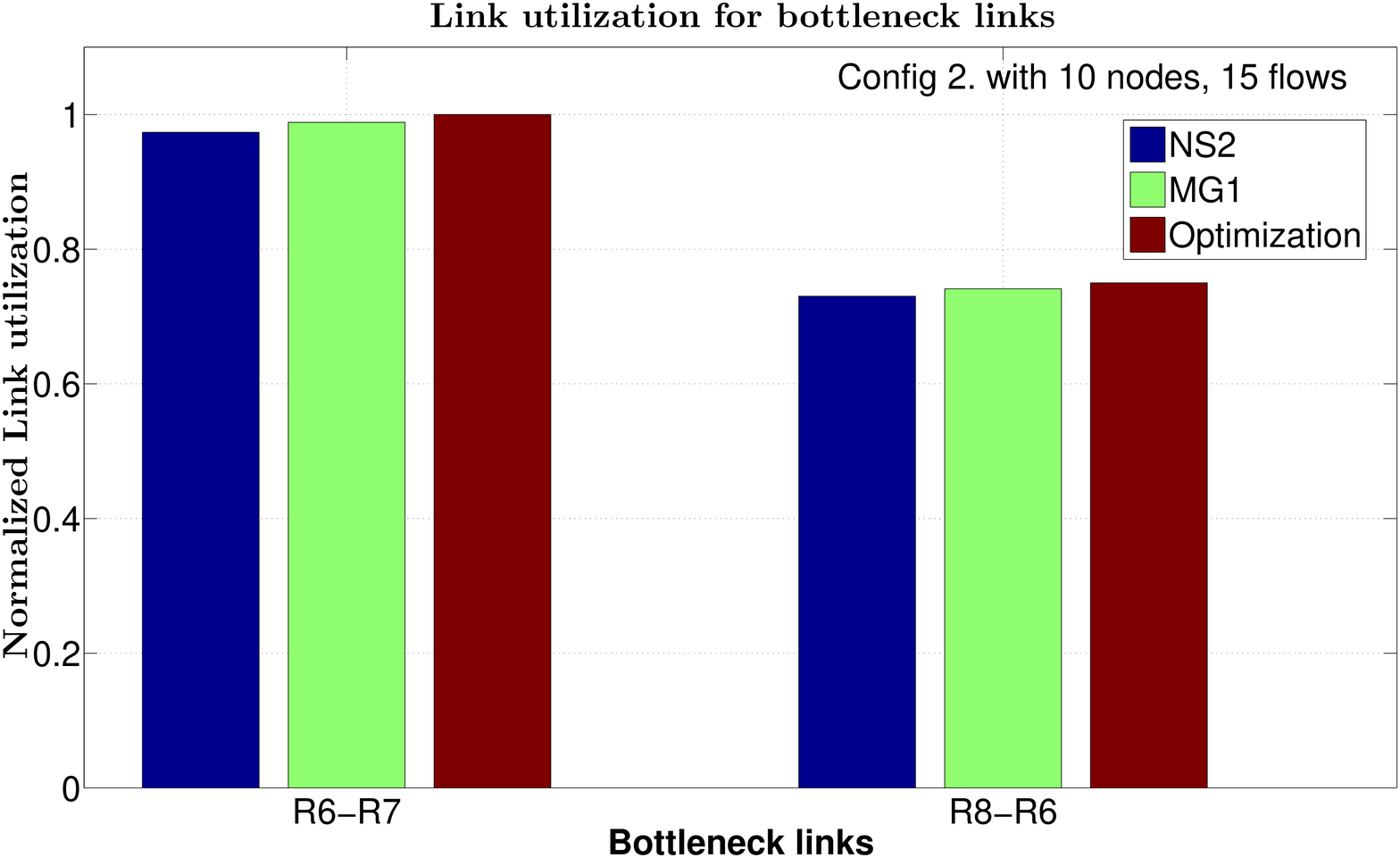}
  \caption{Link Utilization for example in Figure \ref{fig:10N_15F_diag}, config. 2.}
  \label{fig:10N_15F_case9A_util}
\end{figure}

Next we consider the scenario with $12$ routers and $18$ flows with the topology shown in Figure \ref{fig:12N_18F_diag}. We now have $6$ groups of flows with each group consisting of $3$ flows with one flow each of TCP Compound, CUBIC and Reno. In Figures \ref{fig:12N_18F_case7B_2} and \ref{fig:12N_18F_case13B_2} we compare the throughputs obtained by the different flows for two different configurations with network topology given in Figure \ref{fig:12N_18F_diag}. In Figures \ref{fig:12N_18F_case7B_util} and \ref{fig:12N_18F_case13B_util} we compare the link utilization for links with link utilization $> 70 \%$ (in ns2 simulations). The configurations differ in the link speeds (mentioned in Figures \ref{fig:12N_18F_case7B_2} and \ref{fig:12N_18F_case13B_2}) and we denote them as config.3 and config. 4. The parameters used for simulation are mentioned in Figures \ref{fig:12N_18F_case7B_2} and \ref{fig:12N_18F_case13B_2}. The link speeds are mentioned in Mbps and the propagation delays are mentioned in sec.  The packet sizes (in bytes) for the different TCP flows are different and are chosen from $\{600, 1050, 1500\}$. For both the configurations, for the throughput obtained by the different flows, the M/G/1 model and the optimization approach differ from simulations by less than $14\%$. From \ref{fig:12N_18F_case7B_util} and \ref{fig:12N_18F_case13B_util}, we see that both the M/G/1 and the optimization techniques identify the bottleneck links (links with utilization $> 70\%$) correctly. The difference in link utilization, as compared to ns2 simulations, for all the links in the network, is less than $4\%$ for the M/G/1 technique and the optimization technique in both the configurations.

The simulation setting in configurations $3$ and $4$ differ only in the link speeds of the different links with the link speeds in configuration $4$ being greater than equal to the link speeds in configuration $3$. We observe that for flows in group $2$ and group $5$ (flows numbered $4, 5, 6$ and $13, 14, 15$), the TCP CUBIC flow receives the highest throughput in configuration $3$, whereas TCP Compound flow receives the highest throughput in configuration $4$. This behaviour is due to the difference in the queuing at the links in the two different configurations. The packet error rates for flows in group $2$ and group $5$ is $0.001$ and $0.0001$ respectively. For these low per values, TCP Compound flows get higher throughput than TCP CUBIC flows (see Figure \ref{fig:12N_18F_case13B_2}) when there is negligible queuing. However for non-negligible queuing as in configuration $3$, TCP Compound flows behave like Reno while TCP CUBIC flows are more aggressive and get higher throughputs.


\begin{figure}
  \centering
  \includegraphics[scale=0.75]{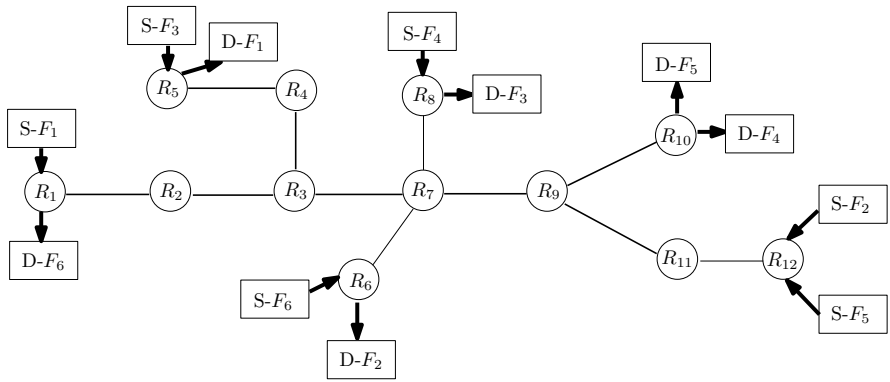}
  \caption{Topology of scenario with 12 routers, 18 flows.}
  \label{fig:12N_18F_diag}
\end{figure}

%
%
%
%

\begin{figure}
  \centering
  \includegraphics[scale=0.23, trim = 80 5 0 0, clip=true]{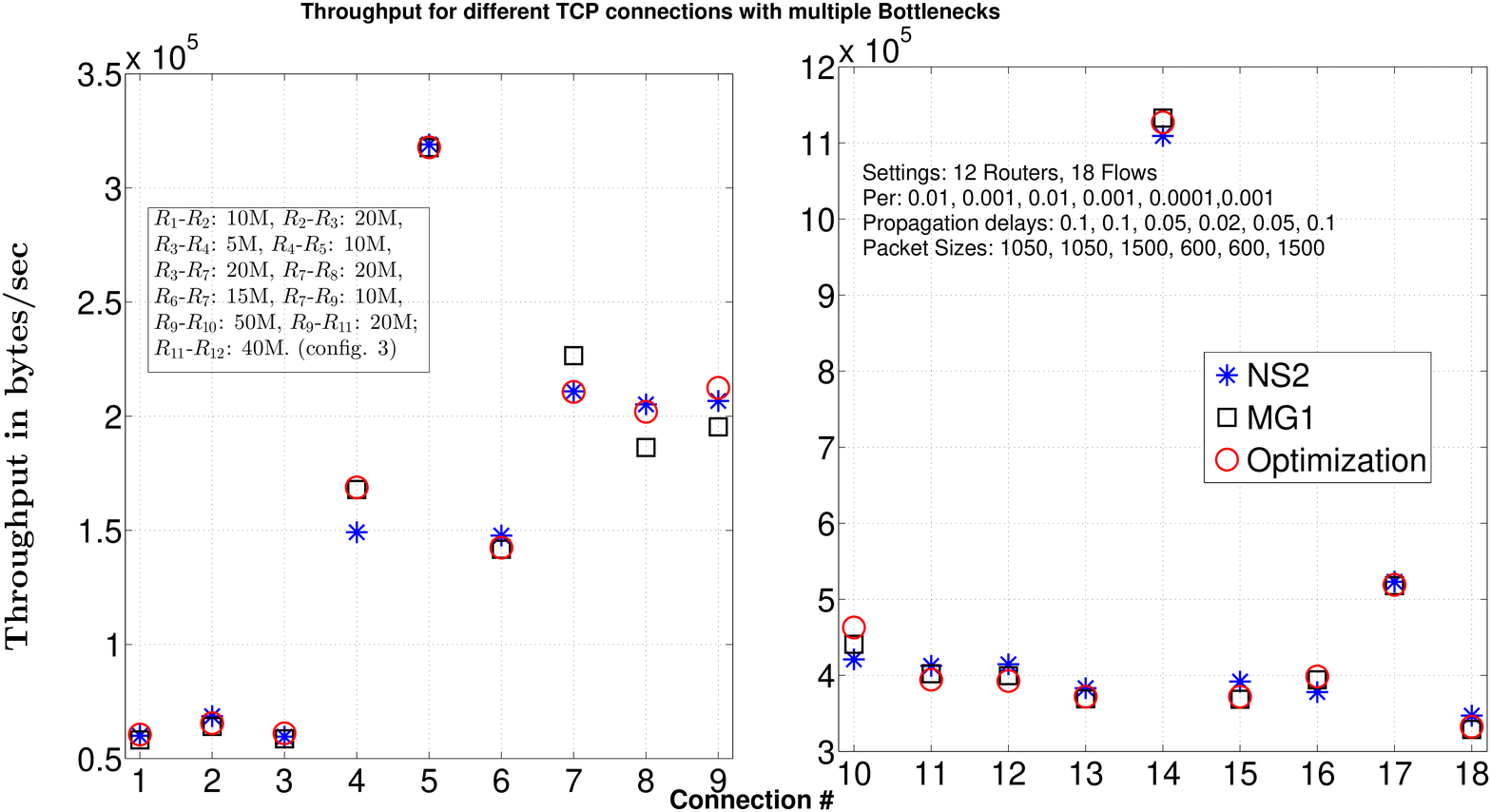}
  \caption{Throughputs for example in Figure \ref{fig:12N_18F_diag}, config. 3.}
  \label{fig:12N_18F_case7B_2}
\end{figure}

\begin{figure}
  \centering
  \includegraphics[scale=0.25, trim = 80 5 140 5, clip=true]{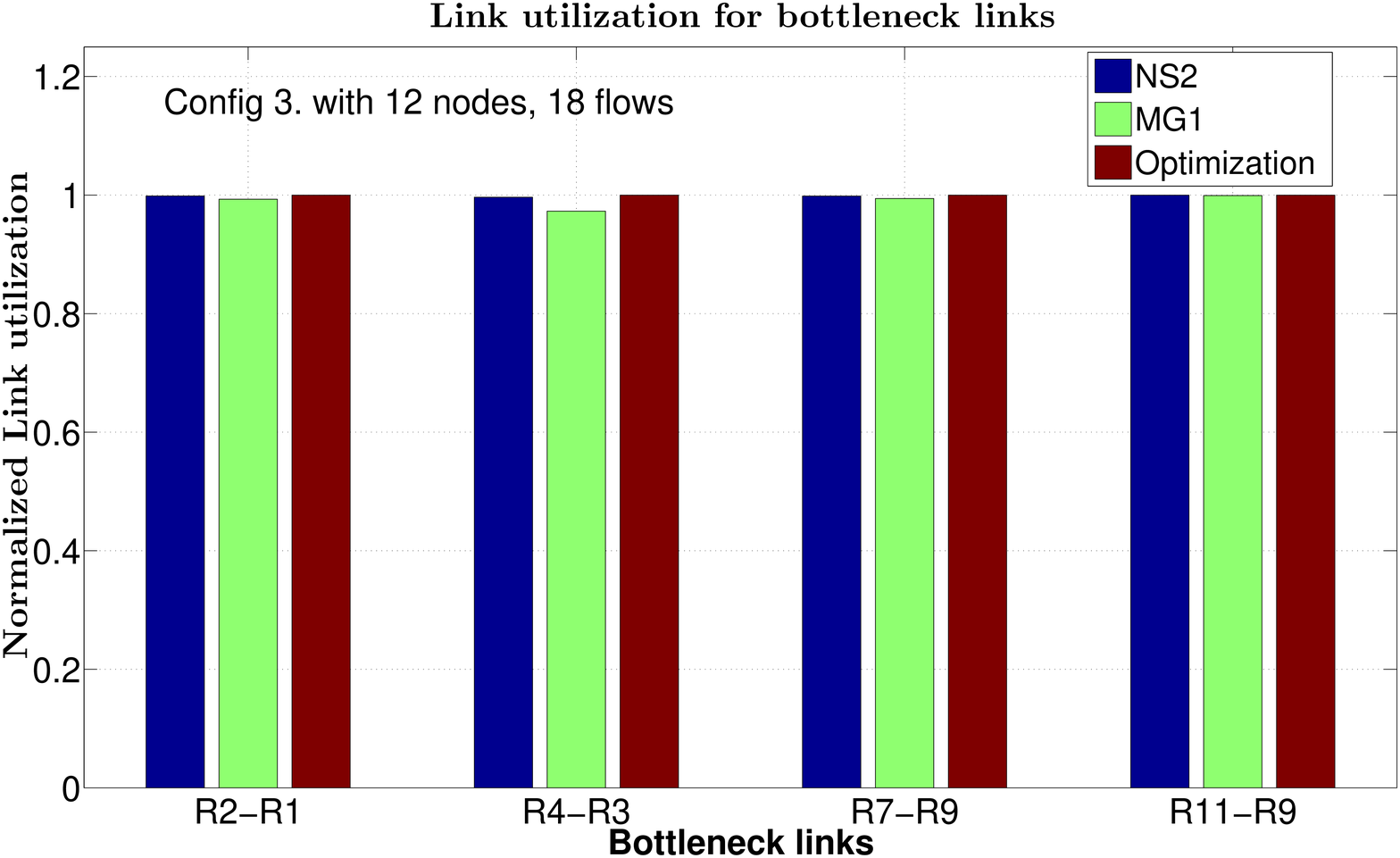}
  \caption{Link Utilization for example in Figure \ref{fig:12N_18F_diag}, config. 3.}
  \label{fig:12N_18F_case7B_util}
\end{figure}

\begin{figure}
  \centering
  \includegraphics[scale=0.25, trim = 80 5 140 5, clip=true]{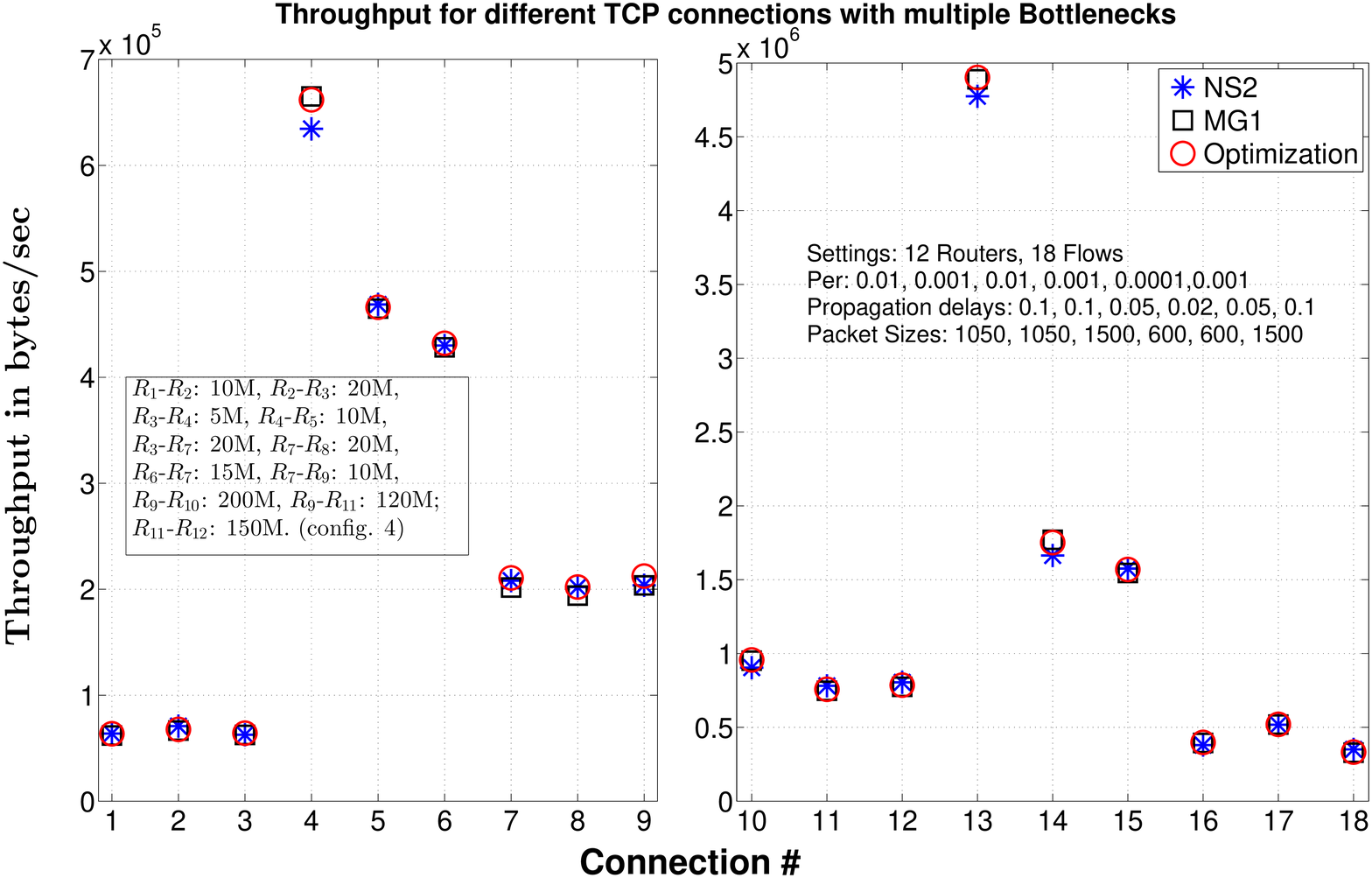}
  \caption{Throughputs for example in Figure \ref{fig:12N_18F_diag}, config. 4.}
  \label{fig:12N_18F_case13B_2}
\end{figure}

\begin{figure}
  \centering
  \includegraphics[scale=0.25, trim = 80 5 140 5, clip=true]{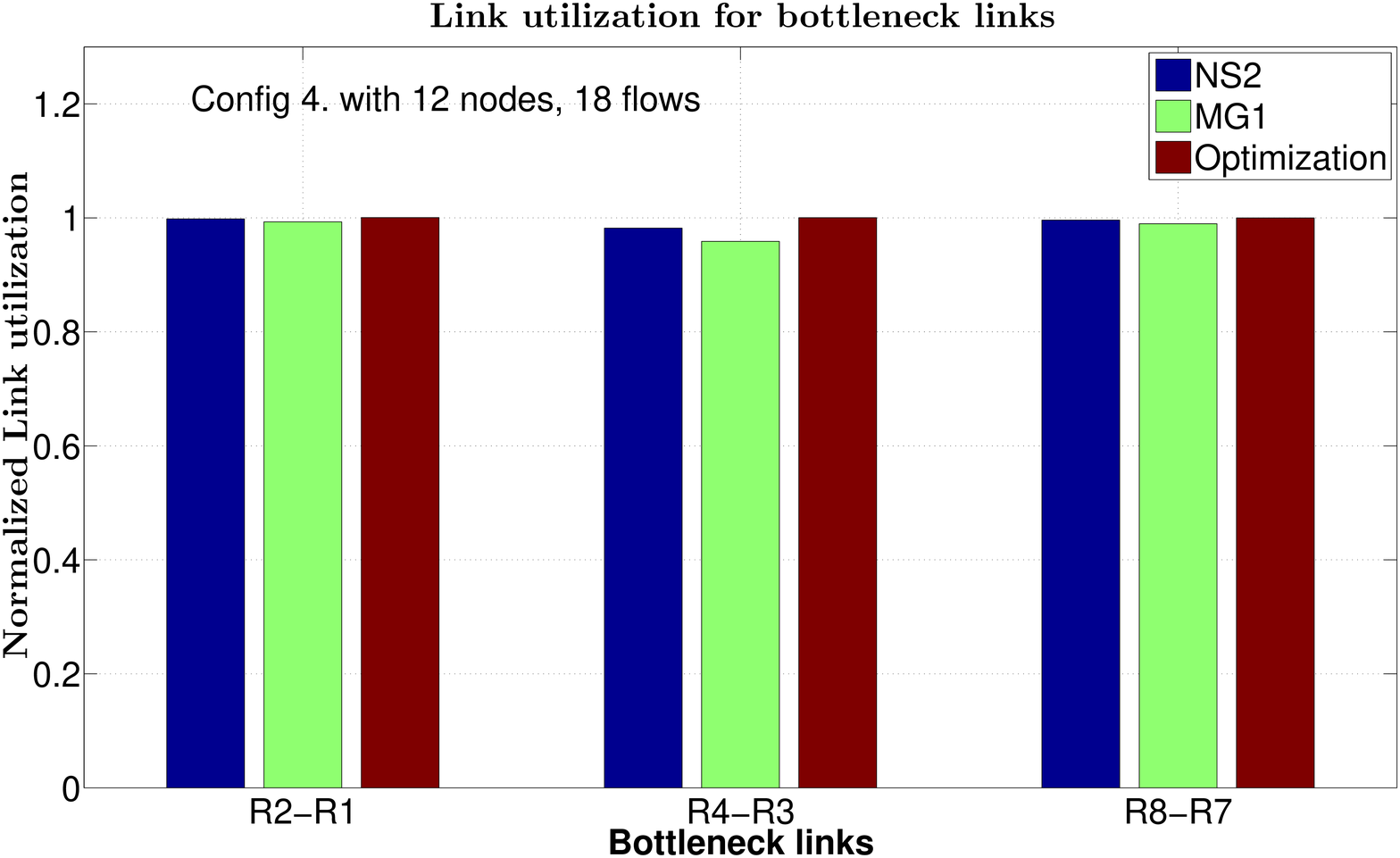}
  \caption{Link Utilization for example in Figure \ref{fig:12N_18F_diag}, config. 4.}
  \label{fig:12N_18F_case13B_util}
\end{figure}

\section{Conclusions}
\label{sec:conclusion}

We have developed Markovian models for computing average window size for a single TCP connection with fixed RTT, i.e., non-negligible queuing with Bernoulli random losses for TCP CUBIC and TCP Compound. For TCP Compound, we compute the average window size for a single TCP flow with non-negligible queuing and random losses. We use two techniques to compute the steady state throughput for multiple TCP flows (which could be using TCP CUBIC, TCP Compound, TCP New Reno) going through a multihop network. The first technique approximates the links by M/G/1 queues. The second technique uses an optimization program whose solution approximates the steady state throughput for the different TCP flows. We compare results obtained from our techniques with ns2 simulation. Our results match well with simulations. 

\bibliographystyle{IEEEtran} 
\bibliography{tcp-references}
\end{document}